\documentclass[journal,10pt]{IEEEtran}

\usepackage{cite}
\usepackage{amsmath,amssymb,amsfonts}
\usepackage{algorithmic}
\usepackage{graphicx}
\usepackage{textcomp}
\usepackage{xcolor}
\usepackage{bm}
\usepackage{diagbox}
\usepackage{multirow}
\usepackage{siunitx}
\usepackage{mathrsfs}

\usepackage[justification=centering]{caption}

\usepackage{graphics,
           psfrag,
           epsfig,
           amsthm,
           cite,
           amssymb,
           url,
           dsfont,
           subfigure,
           algorithm,
           algorithmic,
           balance,
           enumerate,
           color,
           setspace
}

\usepackage{amsmath}               
 {
      \theoremstyle{plain}
      \newtheorem{assumption}{Assumption}

 }

\newtheorem{lemma}{Lemma}

\newtheorem{corollary}{Corollary}

\newtheorem{remark}{Remark}
\usepackage{color,soul}

\newcommand{\eref}[1]{(\ref{#1})}
\newcommand{\sref}[1]{Section~\ref{#1}}

\newcommand{\cref}[1]{Constraint~\ref{#1}}

\newcommand{\tref}[1]{Table~\ref{#1}}

\begin{document}

\title{Constrained Wrapped Least Squares: A Tool for High Accuracy GNSS Attitude Determination}

\author{Xing Liu,
        Tarig Ballal,~\IEEEmembership{Member, IEEE},
        Hui Chen,~\IEEEmembership{Member, IEEE},
        and~Tareq Y. Al-Naffouri,~\IEEEmembership{Senior Member, IEEE}
\thanks{Xing Liu, Tarig Ballal, and Tareq Y. Al-Naffouri are with the Division of Computer, Electrical and Mathematical Sciences, and Engineering, King Abdullah University of Science and Technology, Thuwal 23955-6900, Saudi Arabia (e-mail: xing.liu@kaust.edu.sa; tarig.ahmed@kaust.edu.sa; tareq.alnaffouri@kaust.edu.sa).}

\thanks{Hui Chen is with the Department of Electrical Engineering, Chalmers University of Technology, 41296 Gothenburg, Sweden (e-mail: hui.chen@chalmers.se).}}


\maketitle

\begin{abstract}
Attitude determination is a popular application of Global Navigation Satellite Systems (GNSS). Many methods have been developed to solve the attitude determination problem with different performance offerings. We develop a constrained wrapped least-squares (C-WLS) method for high-accuracy attitude determination. This approach is built on an optimization model that leverages prior information related to the antenna array and the integer nature of the carrier-phase ambiguities in an innovative way. The proposed approach adopts an efficient search strategy to estimate the vehicle's attitude parameters using ambiguous carrier-phase observations directly, without requiring prior carrier-phase ambiguity fixing. The performance of the proposed method is evaluated via simulations and experimentally utilizing data collected using multiple GNSS receivers. The simulation and experimental results demonstrate excellent performance, with the proposed method outperforming the ambiguity function method, the constrained LAMBDA and multivariate constrained LAMBDA methods, three prominent attitude determination algorithms.
\end{abstract}


\IEEEpeerreviewmaketitle

\section{Introduction}
\IEEEPARstart{P}{latform} attitude, crucial information in navigation and vehicle control, is defined as the orientation of a given body frame relative to a reference coordinate \cite{hofmann2012global}. Attitude determination via multiple GNSS antennas mounted on a moving vehicle is an important research field with a wide variety of applications \cite{8713879, 8781924,8972427, 5975219, ChiOct2014,6359888, 8869594}. Recent developments and interest in autonomous systems have resulted in even more attention drawn to this research area \cite{ bijjahalli2017novel,vetrella2016differential, wang2019research,liugnss2019}. 





GNSS attitude determination makes use of simultaneous pseudo-range and carrier-phase observations at multiple antennas from one or more GNSS frequencies \cite{hofmann2007gnss}. It is well known that carrier-phase observations are some orders of magnitude more precise than pseudo-range data \cite{giorgi2013low}. The latter cannot satisfy the needs of high-quality attitude estimation and localization applications. On the other hand, carrier-phase measurements are ambiguous by unknown integer numbers of cycles or wavelengths. Consequently, carrier-phase integer ambiguity resolution is a key and prime difficulty in precise attitude determination. Many ambiguity resolution methods have been proposed over the years, which can be categorized into motion-based methods \cite{crassidis1999global, ChuMay2001, wang2010motion, psiaki2006batch} and search-based methods \cite{Teunissen1995,ziebart2003leo,LiJul2004,PurApr2010, GunHenJul2012, ballal2014gnss}. Motion-based methods have no capacity for real-time solutions due to the requirement of motion information. In contrast, search-based methods can provide instantaneous attitude estimation using the observations from a single time point (or epoch) \cite{LiJul2004}.

Two classes of search-based attitude determination methods can be identified. In the first category, the search is conducted in the (float) attitude domain. Among these methods, the ambiguity function method (AFM) \cite{Counselman1981AFM,han1996improving,yang2016rotation} is the most popular. AFM-based algorithms generally minimize an objective function over a grid of possible attitude angles \cite{LiJul2004}. A considerable number of grid points in the search domain are required to obtain the global optimal, which results in high computational complexity \cite{teunissen2017springer}.

The second category of search-based method applies the search in the ambiguity domain over the space of all potential integer combinations. The most prominent approaches of this class are the LAMBDA method \cite{Teunissen1995} and its various modified versions \cite{6491499,monikes2001modified,chang2005mlambda,Wang2009,GunHenJul2012,teunissen2010integer}. The LAMBDA method searches for the optimal carrier-phase integer ambiguities based on the integer least-squares (ILS) principle with high efficiency. However, it does not consider the prior knowledge of the antenna array configuration. This prior information is beneficial for both enhancing the success rate of ambiguity resolution and the accuracy of attitude estimation \cite{liu2018integrated, Ahmed2020RieOpt, liu2018gnss}. Integer ambiguity resolution under the antenna-geometry constraints is a complex non-convex problem that is very difficult to solve, especially in real-time. The constraint LAMBDA (C-LAMBDA) method has been developed for single-baseline attitude determination, which integrates the baseline length into the optimization model \cite{teunissen2010integer}. In the same spirit, the multivariate constraint LAMBDA (MC-LAMBDA) method was proposed to improve 3-D attitude determination by incorporating information from multiple baselines \cite{teunissen2007general, giorgi2010testing}. These two approaches can significantly enhance the performance of ambiguity resolution at the cost of requiring a more sophisticated search algorithm to meet the nonlinear constraints on the antenna array geometry \cite{park2009integer, giorgi2013low}. Despite the resounding success of these LAMBDA variants, their computational efficiency might be inadequate in challenging environments with significant multipath effects or an adverse view of satellites, such as in urban canyon scenarios.

The main contribution of this paper lies in developing a novel GNSS attitude determination method based on a proposed \emph{constrained wrapped least-squares} (C-WLS) optimization criterion. The designed approach incorporates the antenna array information, the ambiguity integer characteristics, and residual phase constraints into the optimization model. The proposed approach belongs to the search-based category. The method considers the attitude parameters as the only unknown variables and rigorously assimilates all the available constraints. Compared with existing algorithms in the float domain \cite{Counselman1981AFM,han1996improving,yang2016rotation}, the computational complexity of the proposed method is kept low by reducing the search space to  a subset that most likely contains the correct solution. This is done without sacrificing any prior information or constraints. Instead of performing integer ambiguity resolution and attitude estimation in two separate steps, as is done in most traditional approaches \cite{li2002new,cheng2014direct}, we use the ambiguous carrier-phase measurements directly to estimate the attitude information. 
We demonstrate the effectiveness of our approach in different scenarios. As shown by our simulation and experimental results, the main benefit of the proposed method is improving the success rate and computational complexity, especially in challenging situations with a limited number of satellites or large measurement noise.



This paper is organized as follows. \sref{sec:observation} formulates the fundamental GNSS attitude determination problem using the double-difference observation model. \sref{sec:cls} discusses the standard optimization model to solve the attitude determination problem. \sref{sec:promodel} presents the proposed constraint wrapped least-squares method and theoretically analyzes various aspects of the proposed approach. \sref{sec:solution} describes the adopted search strategy to find the attitude solution. In \sref{sec:result}, simulation and experimental results are presented, demonstrating the feasibility of the proposed approach. In \sref{sec:conclusion}, we draw the conclusion of this work.

\section{Background} 
\label{sec:model}
\subsection{Observation Model} 
\label{sec:observation}




For a platform with $\mathcal{A}+1$ GNSS antennas tracking $\mathcal{S}+1$ satellites, the original observations collected at the $a$-th antenna ($a = 0, 1, 2, \cdots, \mathcal{A}$) from the $s$-th satellite's signal ($s = 0, 1, 2, \cdots, \mathcal{S}$) can be modelled as
\begin{equation}
\begin{aligned}
{\rho}_{a}^s &= P _{a}^s + I_{a}^s + T_{a}^s + \frac{c}{\lambda}( {\delta {t_a} - \delta {t^s}} ) + \varepsilon _{a}^s,\\
\psi _{a}^s &= P _{a}^s +  {{N}_{a}^s} - I_{a}^s + T_{a}^s + \frac{c}{\lambda}( {\delta {t_a} - \delta {t^s}} ) + \eta _{a}^s,
\end{aligned}
\label{eq2}
\end{equation}
with
\begin{equation}
P _{a}^s = \frac{1}{\lambda}{\left\| {\mathbf{L}^s - \mathbf{l}_a} \right\|_2}.
\label{eqtion3}
\end{equation}
The variables in (\ref{eq2})--(\ref{eqtion3}) are defined as  follows: $\rho_{a}^s$ and $\psi_{a}^s$ are the pseudo-range and carrier-phase observables (respectively), $P_{a}^s$ denotes the distance between the $s$-th satellite and the $a$-th antenna, $I_{a}^s$ is the ionospheric delay, $T_{a}^s$ is the tropospheric delay, ${N_{a}^s}$ is the integer ambiguity, $\lambda$ represents the wavelength for the given frequency, $c$ is the speed of light, $\delta {t}^s$ and $\delta {t}_{a}$ denote the satellite and receiver clock biases, respectively, $\mathbf{L}^s$ is the satellite location, $\mathbf{l}_a$ is the antenna position, and $\varepsilon_{a}^s$ and $\eta_{a}^s$ account for unmodelled errors and noise. 

For a configuration with multiple antennas, as shown in Fig.~\ref{fig:BodyFrame}, let us denote the baseline direction vector between antenna 0 and antenna $a$ represented in the reference coordinate by ${\mathbf{x}}_a \in {\mathbb{R}^{3}}$. To mitigate the atmospheric effect and satellite clock biases, we calculate the difference of observations at antenna $a$ and the reference antenna, say, antenna 0. This operation is usually denoted as \emph{single difference}, and it results in the the observation model 
\begin{equation}
\begin{aligned}
{\rho}_{a0}^s &= {\rho}_{a}^s - {\rho}_{0}^s 
=P _{a0}^s + \frac{c}{\lambda}\delta {t_{a0}} + \varepsilon _{a0}^s,\\
\psi _{a0}^s &= \psi _{a}^s -\psi _{0}^s 
= P _{a0}^s + {{N}_{a0}^s} + \frac{c}{\lambda}\delta {t_{a0}} + \eta _{a0}^s.
\end{aligned}
\end{equation}
In the attitude determination problem, the baseline length is usually quite short (meter-level or shorter) \cite{zhao2017real}. This makes the atmospheric delays for two antennas highly correlated and can almost be eliminated by the single-difference operation. Since the baseline length is far smaller than the distance from the satellite to the antennas, $P_{a0}^s$ can be represented as \cite{teunissen2012gps}
\begin{equation}
P_{a0}^s =P_{a}^s - P _{0}^s = \frac{1}{\lambda} \left({\bf{h}} ^s \right)^\text{T} {\mathbf{x}}_a,
\end{equation}
where ${\bf{h}} ^s \in {\mathbb{R}^{3}}$ is the unit line-of-sight vector between the receiver and satellite $s$, and ${\text{T}}$ denotes the matrix transpose. 


Another difference operation is carried out on the single-difference observations over pairs of satellites to cancel out the receiver clock bias. By designating satellite 0 as a reference, the difference process results in the so-called \emph{double-difference} model, which is given by
\begin{equation}
\begin{aligned}
\rho_{a0}^{s0} &= {\rho}_{a0}^{s} - {\rho}_{a0}^{0} 
={\frac{1}{\lambda}\!\left({\bf{h}} ^s-{\bf{h}}^{0} \right)^\text{T} {\mathbf{x}}_a} + \varepsilon_{a0}^{s0},\\
 \psi _{a0}^{s0} &=  \psi _{a0}^{s} - \psi _{a0}^{0}
= {\frac{1}{\lambda} \left({\bf{h}} ^s - {\bf{h}} ^{0}\right)^\text{T} {\mathbf{x}}_a} +  {{N}} _{a0}^{s0} + \eta _{a0}^{s0}.
 \end{aligned}
 \label{dd:eq2}
\end{equation}
The double-difference operation eliminates most of the common errors such as atmospheric delay, ephemeris errors, and clock errors. 
\begin{figure}[tbp]
\centering
\includegraphics[width=0.31\textwidth]{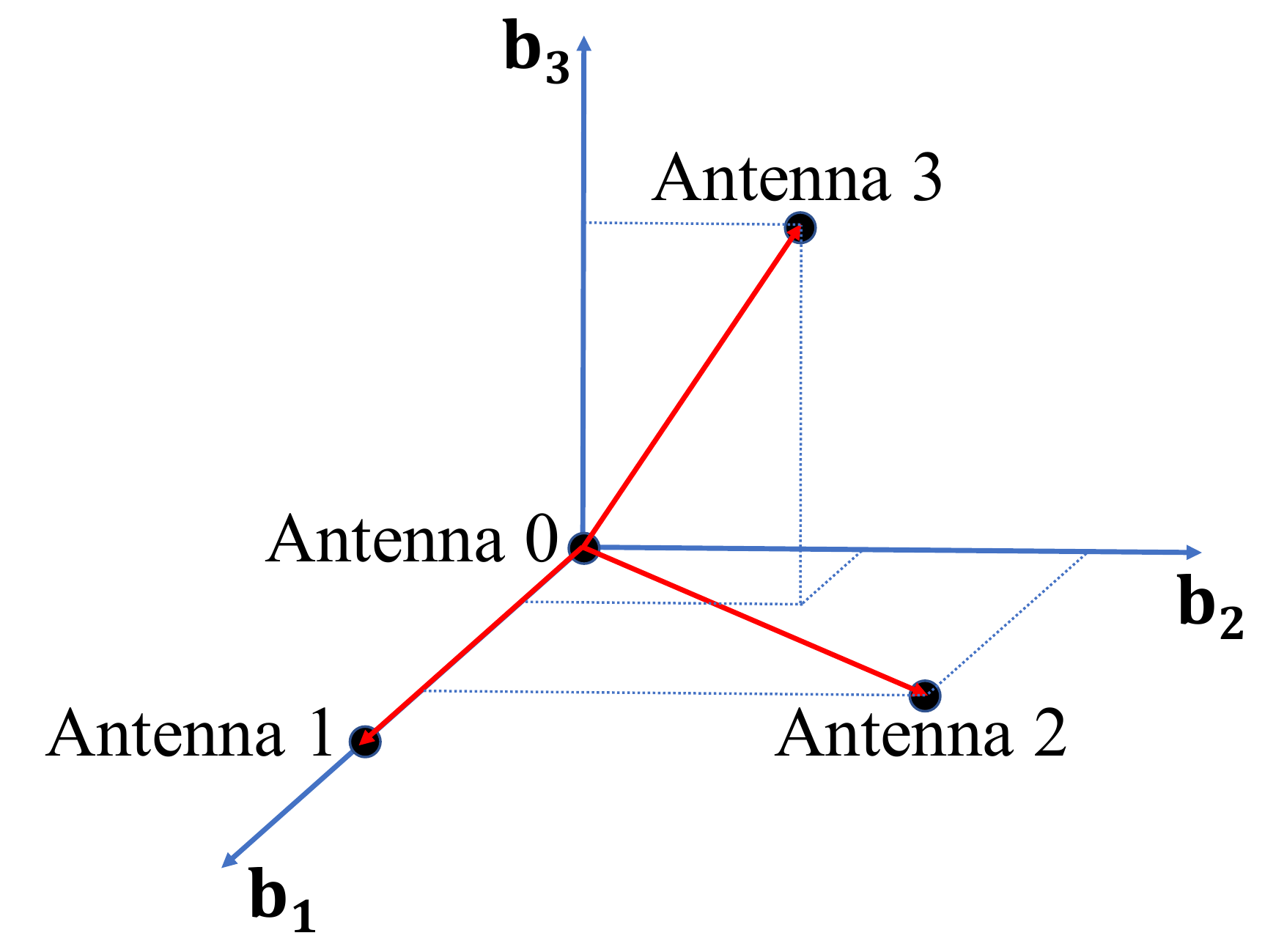}
\caption{GNSS antennas in the body frame.}
\label{fig:BodyFrame}
\end{figure}

For $a$-th baseline (the baseline between antenna 0 and antenna $a$), we can arrange the double-difference observations from various satellites in the vectors
\begin{equation}
{\bm{\rho }}_a = {\left[\!\! {\begin{array}{*{20}{c}}
  {{\rho }_{a0}^{10}}& \cdots &{{\rho}_{a0}^{\mathcal{S}0}}
\end{array}} \!\! \right]^{\text{T}}}, \quad
{\bm{\psi}}_a  = {\left[\!\! {\begin{array}{*{20}{c}}
  {\psi _{a0}^{10}}& \cdots &{{\psi}_{a0}^{\mathcal{S}0}} 
\end{array}} \!\! \right]^{\text{T}}}.
\label{dd:eq4}
\end{equation}
By arranging the other parameters in (\ref{dd:eq2}) in a vector or a matrix to match (\ref{dd:eq4}), the standard double-difference model is given by the following linear observation equations:
\begin{equation}
\begin{aligned}
{\bm{\rho}}_a &= {\mathbf{H}}\mathbf{x}_a + \bm{\epsilon}_a,\\
  {\bm{\psi}}_a &= {\mathbf{H}}{\mathbf{x}_a} + \mathbf{N}_a + \bm{\eta}_a,
  \end{aligned}
\label{eq:model2}
\end{equation}
where $\bm{\epsilon}_a, \bm{\eta}_a \in {\mathbb{R}^{\mathcal{S}}}$ are the unmodeled error and noise vectors, $\mathbf{N}_a \in {\mathbb{Z}^{\mathcal{S}}}$ is the unknown integer ambiguity vector, and ${\mathbf{H}} \in {\mathbb{R}^{\mathcal{S} \times 3}}$ is a design (non-singular) matrix with the form
\begin{equation}
{\bf{H}} = \frac{1}{\lambda}\left[ \!\! {\begin{array}{*{20}{c}}
\left({{\bf{h}}^1} - {{\bf{h}}^{0}}\right)^\text{T}\\
 \vdots \\
\left({{\bf{h}}^{\mathcal{S}}} - {{\bf{h}}^{0}}\right)^\text{T}
\end{array}}\!\! \right] = \left[\!\! {\begin{array}{*{20}{c}}
{{\bf{\tilde h}}_1}\\
 \vdots \\
{{\bf{\tilde h}}_{\mathcal{S}}}
\end{array}} \!\! \right].
\label{eq:H}
\end{equation}


The process of estimating $\mathbf{x}_a$ based on \eqref{eq:model2} is referred to as \emph{single-baseline attitude determination}. The baseline orientation given by the vector $\mathbf{x}_a$ only partly characterizes the platform's attitude. To fully describe the platform's attitude, two or more non-collinear baselines are required. By combining the measurement vectors for multiple baselines in a matrix, the 3-D GNSS attitude observation model is given by
\begin{equation}
\begin{aligned}
{\mathbf{P}} &= {\mathbf{H}}\mathbf{X} + \bm{\Xi},\\
  {\bm{\Psi}} &= {\mathbf{H}}{\mathbf{X}} + \mathbf{N} + \bm{\Pi},
  \end{aligned}
  \label{eq:model4}
\end{equation}
with
\[{\mathbf{P}} \!=\! \left[ \!\!\!{\begin{array}{*{20}{c}}
  {{\bm{\rho}}_{_1}} \!\!\! &\cdots \!\!\! &{\bm{\rho}}_{_\mathcal{A}} 
\end{array}} \!\!\!\right],
{\bm{\Psi}} \!= \!\left[ \!\!\!{\begin{array}{*{20}{c}}
  {{{\bm{\psi}}_{_1}}}& \!\!\!\cdots \!\!\! &{{{\bm{\psi}}_{_\mathcal{A}}}} 
\end{array}} \!\!\!\right],
{{\mathbf{X}}} \!= \! \left[\!\!\! {\begin{array}{*{20}{c}}
  {{{\mathbf{x}}_{_1}}}& \!\!\!\cdots \!\!\!&{{{\mathbf{x}}_{_\mathcal{A}}}} 
\end{array}}\!\!\! \right],\]
where $\mathbf{N} $, $\boldsymbol{\Xi}$ and $\boldsymbol{\Pi}$ are the corresponding matrix forms of the integer ambiguities and error variables. Note that the pseudo-range and carrier-phase are assumed to be uncorrelated, whose covariance matrices are given by
\[{{\mathbf{Q}}_{{\mathbf{P}}}} = \operatorname{cov}\!\left[ \operatorname{vec}\!\left({\mathbf{P}}\right)\right] = \operatorname{E}\!\!\left[ \operatorname{vec}({\mathbf{P}} \!-\! \operatorname{E}({\mathbf{P}}) )\left[\operatorname{vec}({\mathbf{P}} \!-\! \operatorname{E}({\mathbf{P}}) )\right]^{\text{T}}\right],\]
\[{{\mathbf{Q}}_{{\bm{\Psi }}}} = \operatorname{cov}\!\left[\operatorname{vec}\!\left( {\bm{\Psi}}\right)\right] = \operatorname{E}\!\!\left[ \operatorname{vec}(\bm{\Psi } \!-\! \operatorname{E}(\bm{\Psi }) )\left[\operatorname{vec}(\bm{\Psi } \!-\! \operatorname{E}(\bm{\Psi }) )\right]^{\text{T}}\right],\]
\[{{\mathbf{Q}}} = \operatorname{cov}\!\!\left( \!\left[\!\!\! {\begin{array}{*{20}{c}}
  \operatorname{vec}\!\left({\bm{\Psi}}\right) \\ 
  \operatorname{vec}\!\left({\bf{P}}\right)
\end{array}} \!\!\!\right]\!\right) = \left[ \!\!\!{\begin{array}{*{20}{c}}
  {\mathbf{Q}_{{\bm{\Psi }}}}&\mathbf{O} \\ 
  \mathbf{O}&{\mathbf{Q}_{{\bf{P}}}} 
\end{array}}\!\!\!\right],\]
where $\operatorname{vec}(\cdot)$ is the vectorization operation, $\operatorname{E}( \cdot )$ denotes the expectation operator, and $\mathbf{O}$ is a zero matrix of dimension $(\mathcal{A}\mathcal{S}) \times (\mathcal{A}\mathcal{S})$.

\subsection{The Constrained Integer Least-Squares Solution}
\label{sec:cls}
In this subsection, we summarize the standard technique for solving the attitude determination problem based on the observation model (\ref{eq:model4}). In the setup under consideration, the antennas are firmly fixed on a rigid platform such that their coordinates in the body frame can be precisely measured. Hence, the baseline coordinate matrix expressed in the body-frame coordinate system, ${\mathbf{X}_b}$, is known. The matrix ${\mathbf{X}_b}$ represents the translation of the matrix ${\mathbf{X}}$ in (\ref{eq:model4}) from the chosen reference coordinate system to the body frame. For both ${\mathbf{X}}$ and ${\mathbf{X}_b}$, all the nonlinear geometrical constraints of the baseline lengths and the baseline relative orientations are satisfied. A rotation matrix links these two baseline matrices and maintains the geometrical constraints intact, that is
\begin{equation}
{{\mathbf{X}}} = {{\mathbf{R}}}{{\mathbf{X}_b}},
\end{equation}
where ${\mathbf{R}} \in {\mathbb{O}^{3 \times q}}$ is an orthogonal matrix, ${\mathbf{R}}^{\text{T}}{\mathbf{R}} = {\mathbf{I}_q}$, ${\mathbf{I}_q}$ is an identity matrix of dimension $q$, and $q = \operatorname{min}(3,\mathcal{A})$. The essence of the 3-D attitude determination is to estimate the rotation matrix $\mathbf{R}$, which represents the orientation of the body frame relative to the reference frame. 

Depending on the value of $\mathcal{A}$ and $q$, we identify three different cases. First, for $\mathcal{A} = 1$, $q = 1$, ${\mathbf{R}}$ is a column vector that contains only two independent entries (the third entry is determined by the unit-norm requirement). In this case, ${\mathbf{X}_b}$ degenerates to a scalar value
\begin{equation}
{{\bf{X}}_b} = {x_{_{11}}}.
\end{equation}
For $\mathcal{A} = 2$, $q = 2$, the baseline coordinates in the body frame can be defined as
\begin{equation}
{{\bf{X}}_b} = \left[\!\! {\begin{array}{*{20}{c}}
{{x_{_{11}}}}&{{x_{_{21}}}}\\
0&{{x_{_{22}}}}
\end{array}}\!\! \right].
\end{equation}
Finally, for $\mathcal{A} \geq 3$, $q = 3$, the baseline coordinates in the body frame take the form
\begin{equation}
{{\bf{X}}_b} = \left[ \!\!{\begin{array}{*{20}{c}}
{{x_{_{11}}}}&{{x_{_{21}}}}&{{x_{_{31}}}}& \cdots &{{x_{_{\mathcal{A}1}}}}\\
0&{{x_{_{22}}}}&{{x_{_{32}}}}& \cdots &{{x_{_{\mathcal{A}2}}}}\\
0&0&{{x_{_{33}}}}& \cdots &{{x_{_{\mathcal{A}3}}}}
\end{array}}\!\! \right].
\end{equation}

The estimation of the unknown integer ambiguities and the rotation matrix can be formulated as the constrained least-squares (C-LS) optimization \cite{teunissen2010integer}
\begin{equation}
\mathop {\min }\limits_{{\mathbf{R}} \in {\mathbb{O}^{3 \times q}}, {\mathbf{N}} \in {\mathbb{Z}^{\mathcal{S} \times\mathcal{A}}}} {\text{ }}\left\| {\operatorname{vec}\left( {{\mathbf{Y}} - {\mathbf{AR}}{{\mathbf{X}}_b} - {\mathbf{BN}}} \right)} \right\|_{\mathbf{Q}_{\mathbf{Y}}^{-1}}^2,
\label{eq:cls}
\end{equation}
where
\[{\mathbf{Y}} \triangleq \left[ {\begin{array}{*{20}{c}}
  {\bm{\Psi}} \\ 
  {\bf{P}}
\end{array}} \right],
{\mathbf{A}} \triangleq \left[ {\begin{array}{*{20}{c}}
  {\mathbf{H}} \\ 
  {\mathbf{H}} 
\end{array}} \right],
{\mathbf{B}} \triangleq \left[ {\begin{array}{*{20}{c}}
  \mathbf{I}_\mathcal{S} \\ 
  \mathbf{O}
\end{array}} \right].\]
The design matrices ${\mathbf{A}}$ and ${\mathbf{B}}$ link the observation vector to the unknown parameters, ${\mathbf{Q}_{\mathbf{Y}}}$ is the variance-covariance matrix of $\operatorname{vec}\left({\mathbf{Y}}\right)$, and $\left\| (\cdot) \right\|_{{\mathbf{Q}}_{\mathbf{Y}}^{-1}}^2 = (\cdot)^{\text{T}}{\mathbf{Q}}_{\mathbf{Y}}^{-1}(\cdot)$.

The optimization in (\ref{eq:cls}) is known as the \emph{constrained (mixed) integer least-squares} (C-ILS) problem. The C-ILS model handles two types of constraints: the integer constraints of the carrier-phase ambiguities, and the orthogonality constraints of the rotation matrix. This is a non-convex optimization due to the integer property of the carrier-phase ambiguities and the orthogonality constraints. A renowned approach to solve (\ref{eq:cls}) is the MC-LAMBDA method, which makes use of the constrained integer least-squares theory along with a search-and-shrink or search-and-expand strategy \cite{teunissen2010integer}. The MC-LAMBDA method seeks the optimal carrier-phase ambiguities in the integer domain, which may lead to high complexity in some scenarios.


\section{The Proposed Constrained Wrapped Least-Squares Method}
\label{sec:promodel}
To avoid the complex process of resolving the carrier-phase integer ambiguities, we propose a novel optimization model to estimate the rotation matrix directly and recover the unambiguous phase as a by-product. This model leads to a very efficient approach to solve the attitude determination problem. We will start by introducing the proposed model using only carrier-phase observations. Subsequently, we extend the model to include pseudo-range measurements.

To strengthen the model (\ref{eq:cls}), additional constraints (prior information) will be exploited. Given carrier-phase observations, the unknown matrices ${\mathbf{R}}$ and ${\mathbf{N}}$ can be estimated by applying the C-ILS optimization
\begin{equation}
\mathop {\min }\limits_{{\mathbf{R}} \in {\mathbb{O}^{3 \times q}},{\mathbf{N}} \in {\mathbb{Z}^{\mathcal{S} \times\mathcal{A}}}} \left\| \operatorname{vec} \! \left[\bm{\xi} ( {\mathbf{R}}, {\mathbf{N}} ) \right] \right\|_{{\mathbf{Q}}^{-1}_{{\bm{\Psi}}}}^2,
\label{eq:eq23}
\end{equation}
where $\bm{\xi}( {\mathbf{R}}, {\mathbf{N}} )$ is the residual phase error defined as
\begin{equation}
\bm{\xi} ( {\mathbf{R}}, {\mathbf{N}} ) \triangleq {{{\bm{\Psi }} - {\mathbf{H}} {\mathbf{R}}{\mathbf{X}_b}} - {\mathbf{N}}}.
\end{equation}
Let $N_{ij}$ and $\xi_{ij}( {\mathbf{R}}, N_{ij} )$ be the entries in the $i$-th row and $j$-th column of ${\mathbf{N}}$ and $\bm{\xi} ( {\mathbf{R}}, {\mathbf{N}} )$, respectively. Considering the observation error distribution, we can apply an upper-bound to restrict the residual phase, i.e.,
\begin{equation}
\left| \xi_{ij}( {\mathbf{R}}, N_{ij} ) \right| \leqslant \delta, \!\quad i = 1, 2, \cdots, \mathcal{S}; \! \quad j= 1, 2, \cdots, \mathcal{A};
\label{eq:residual}
\end{equation}
where $\left| \cdot \right|$ denotes the absolute value, and $\delta$ is an upper bound that can be related to the noise level. In what follows we will state our main assumption concerning this upper bound. This assumption significantly simplifies the process of estimating the rotation matrix, as will be shown subsequently.

\begin{assumption} \label{as:1}
Double-difference carrier-phase observation noise is confined to an interval bounded by minus and plus a half wavelength. That is
\begin{equation}
\left| \eta _{a0}^{s0} \right| \leqslant \frac{1}{2}, \!\quad s = 1, 2, \cdots, \mathcal{S}; \! \quad a= 1, 2, \cdots, \mathcal{A}.
\label{eq:noise}
\end{equation}
\end{assumption}

This assumption is essential to the development of the proposed attitude determination method. Therefore, we examine the validity of this assumption before proceeding further. We analyze the assumption from two points of view, as follows.
\begin{itemize}
\item \textbf{Measurement noise level}

To analyze the double-difference observation noise level, we apply the Gaussian model to the original carrier-phase measurement noise, that is, we assume the noise of the undifferenced phase follows a Gaussian distribution $\mathcal{N}(0,\sigma_0^2)$. By the rules of error propagation, the noise of the double-difference phase also follows a Gaussian distribution $\mathcal{N}(0,\sigma^2)$ with $\sigma = 2\sigma_0$. The possibility that the double-difference carrier phase noise exceeds half of a cycle is 
\begin{equation}
P_{\frac{1}{2}} = 2 \mathcal{Q} \! \left( {\frac{1}{2}} \right),
\end{equation}
where $\mathcal{Q} (\cdot )$ is the Q-function. As reported in \cite{hofmann2012global}, the carrier-phase observations can be measured to better than 0.01 wavelength. By setting $\sigma_0 = 0.01$ and $\sigma = 0.02$, we find out that $P_{\frac{1}{2}} \approx 0$. Thus, it is generally reasonable to regard $P_{\frac{1}{2}}$ as small enough to be ignored.

\vspace{2mm} 

\item \textbf{Phase observation nature}
  
In the context of attitude determination, only the fractional part of the double-difference phase data is meaningful. This particular property of phase measurements is due to the presence of integer ambiguities. Based on (\ref{eq:model2}), for two different observed noise vectors $\bm{\eta}_a$ and $\bm{\eta}_a  + \Delta \!{\bf{N}}_a$, $\Delta \!{\bf{N}}_a \in \mathbb{Z}^{\mathcal{S}}$, the double-difference phase observations ${\bm{\psi}_a}$ will have identical fractional parts. This results in the same estimations of the rotation matrix and the unambiguous double-difference carrier phase based on (\ref{eq:cls}). The integer $\Delta \!{\bf{N}}_a$ will be absorbed in the integer ambiguity estimation. Therefore, it is impossible to distinguish whether the phase noise is $\bm{\eta}_a$ or $\bm{\eta}_a  + \Delta \!{\bf{N}}_a$. Given that achieving the optimal attitude information is the real goal, we can make an assumption about the integer part of phase noise to restrict the possible solutions in a way that is beneficial in improving the estimation of the rotation matrix. Hence, we can select the integer to minimize the absolute value of phase noise, which will lead the phase noise to fall in $\left[-1/2,1/2\right]$; that is, we consider only the estimations corresponding to observations with the lowest possible noise. In other words, double-difference carrier-phase observation noise can be assumed to be no greater than a half wavelength.
\end{itemize}

Based on Assumption~\ref{as:1}, we obtain the following constraint on the residual phase:  
\begin{equation}
\left| \xi_{ij}( {\mathbf{R}}, N_{ij} ) \right| \leqslant \frac{1}{2}, \! \quad i\!= \!1, 2, \cdots, \mathcal{S};\! \quad j \!= \! 1, 2, \cdots, \mathcal{A}.
\label{res:eq1}
\end{equation}
Note that $\left| \xi_{ij}( {\mathbf{R}}, N_{ij} ) \right| = 1/2$ will result in $N_{ij}$ having multiple solutions for the same ${\mathbf{R}}$. To facilitate obtaining an interesting form of the proposed method, we introduce a minor technical trick to modify (\ref{res:eq1}) by limiting the range of $\xi_{ij}( {\mathbf{R}}, N_{ij} )$ to $\left(-1/2,1/2\right]$, that is
\begin{equation}
-\frac{1}{2} \!<\! \xi_{ij}( {\mathbf{R}}, N_{ij} ) \!\leqslant \! \frac{1}{2}, \! \! \quad \!i\!= \!1, 2, \cdots \!, \mathcal{S};\!\! \quad \! j \!= \! 1, 2, \cdots \!, \mathcal{A}.
\label{res:eq2}
\end{equation}
Combining (\ref{eq:eq23}) with (\ref{res:eq2}), we obtain the following version of the minimization problem for attitude determination:
\begin{subequations}
\label{eq15}
\begin{align}
\label{eq15a}
&\mathop {\min }\limits_{{\mathbf{R}} \in {\mathbb{O}^{3 \times q}},{\mathbf{N}} \in {\mathbb{Z}^{\mathcal{S} \times\mathcal{A}}}}  \left\| \operatorname{vec} \! \left[\bm{\xi} \!( {\mathbf{R}}, {\mathbf{N}} ) \right] \right\|_{{\mathbf{Q}}^{-1}_{{\bm{\Psi}}}}^2, \\
\label{eq15b}
  &\text{s.t.}  -\frac{1}{2} < \xi_{ij}( {\mathbf{R}}, N_{ij} ) \leqslant \frac{1}{2},
\end{align}
\end{subequations}
\[i = 1, 2, \cdots, \mathcal{S};\quad j= 1, 2, \cdots, \mathcal{A}.\]
We can rewrite (\ref{eq15a}) in an alternative form as
\begin{equation}
\mathop {\min }\limits_{{\mathbf{R}} \in {\mathbb{O}^{3 \times q}}} \left( \mathop {\min }\limits_{{\mathbf{N}} \in {\mathbb{Z}^{\mathcal{S} \times\mathcal{A}}}} \left\| \operatorname{vec} \! \left[\bm{\xi} \! \left( {\mathbf{R}}, {\mathbf{N}} \right) \right] \right\|_{{\mathbf{Q}}^{-1}_{{\bm{\Psi }}}}^2 \right).
\label{eq:eq22}
\end{equation}
Given any value of $\mathbf{R}$, the corresponding value of ${\mathbf{N}}$ is
\begin{equation}
{\mathbf{\tilde N}} ( {\mathbf{R}} ) = \operatorname{round}( {{\bm{\Psi }} - {\mathbf{H}} {\mathbf{R}} {\mathbf{X}_b}} ).
\label{eq:eq24}
\end{equation}
with $\operatorname{round}( \cdot )$ being a special rounding function that works exactly like a standard rounding function except that for ${N \in {\mathbb{Z}}}$, $\operatorname{round}( N + 0.5 ) = N$.
Using \eqref{eq:eq24}, \eqref{eq:eq22} can be rewritten as a single optimization
\begin{equation}
\mathop {\min }\limits_{{\mathbf{R}} \in {\mathbb{O}^{3 \times q}}} \left\| \operatorname{vec} \! \left[{{\bm{\Psi }} - {\mathbf{H}} {\mathbf{R}} {\mathbf{X}_b}} -  {\mathbf{\tilde N}} ( {\mathbf{R}} ) \right] \right\|_{{\mathbf{Q}}^{-1}_{{\bm{\Psi }}}}^2.
\label{eq:eq26}
\end{equation}
Based on (\ref{eq:eq26}), a more compact form of the attitude determination problem formulated by \eqref{eq15a} and \eqref{eq15b} is given in the following lemma.


\begin{lemma} \label{lem1}
Based on Assumption~\ref{as:1}, the rotation matrix ${\mathbf{R}}$ can be estimated by solving the minimization
\begin{equation}
\mathop {\min }\limits_{{\mathbf{R}} \in {\mathbb{O}^{3 \times q}}} \left\| 
  {\operatorname{wrap} \! \left[ \operatorname{vec} \! \left({{\bm{\Psi }} - {\mathbf{H}} {\mathbf{R}} {\mathbf{X}_b}} \right) \right]}  \right\|_{{\mathbf{Q}^{-1}_{\bm{\Psi }}}}^2, 
\label{eq3}
\end{equation}
where ${\operatorname{wrap}}( \cdot )$ is defined as
\begin{equation*}
{\operatorname{wrap}}(\cdot) = (\cdot) - {\operatorname{round}}( \cdot).
\label{eq:eq15}
\end{equation*}
\end{lemma}

We refer to (\ref{eq3}) as the \emph{constrained wrapped least-squares} (C-WLS) problem. Here, the rotation matrix ${\mathbf{R}}$ is the only unknown. Once ${\mathbf{R}}$ is known, the corresponding integer vector ${\mathbf{N}}$ can be determined from (\ref{eq:eq24}).
\begin{figure}[tbp]
\centering 
\subfigure[] { \label{fig:a} 
\includegraphics[width=0.45\columnwidth]{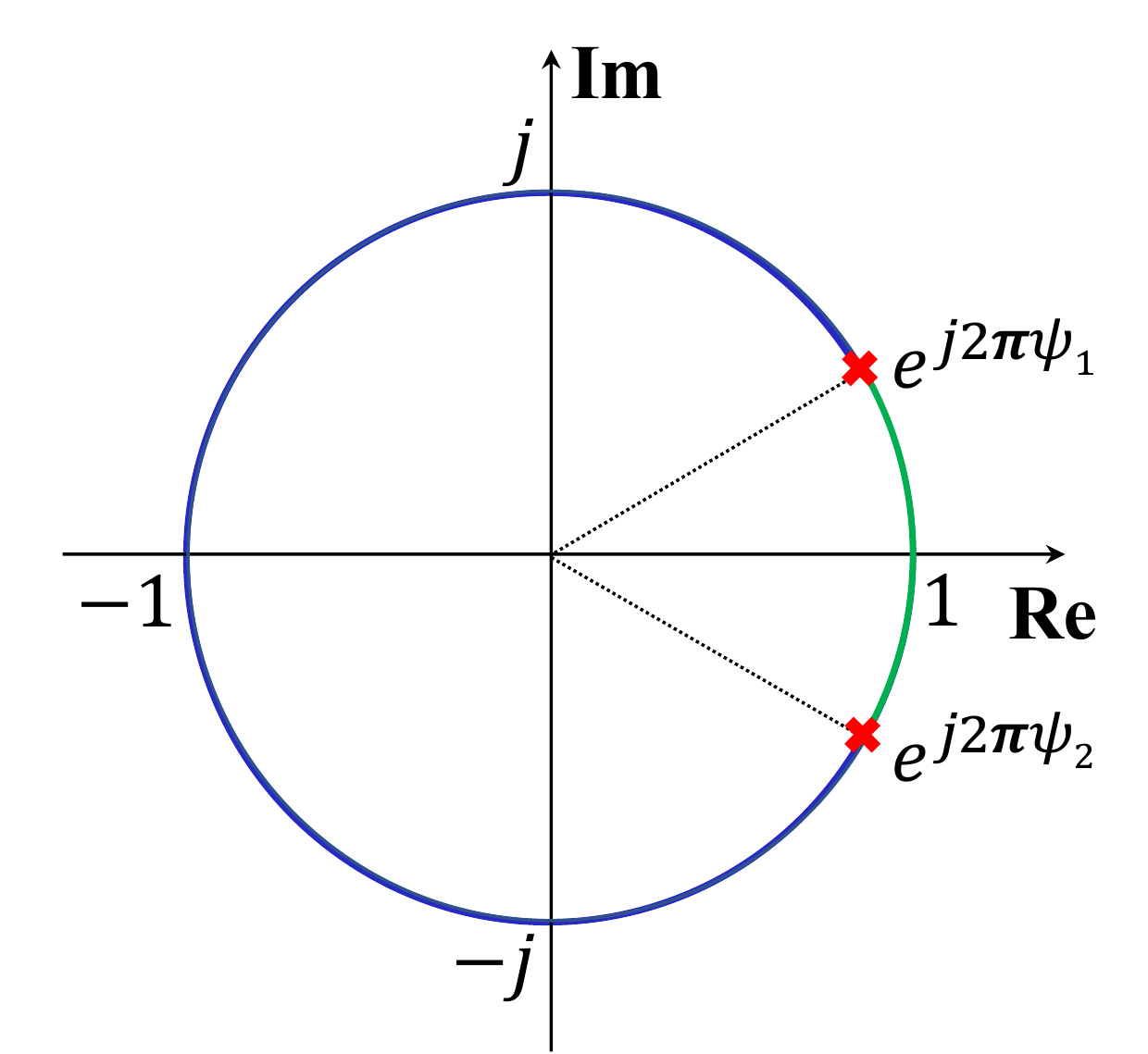}
} 
\subfigure[] { \label{fig:b} 
\includegraphics[width=0.48\columnwidth]{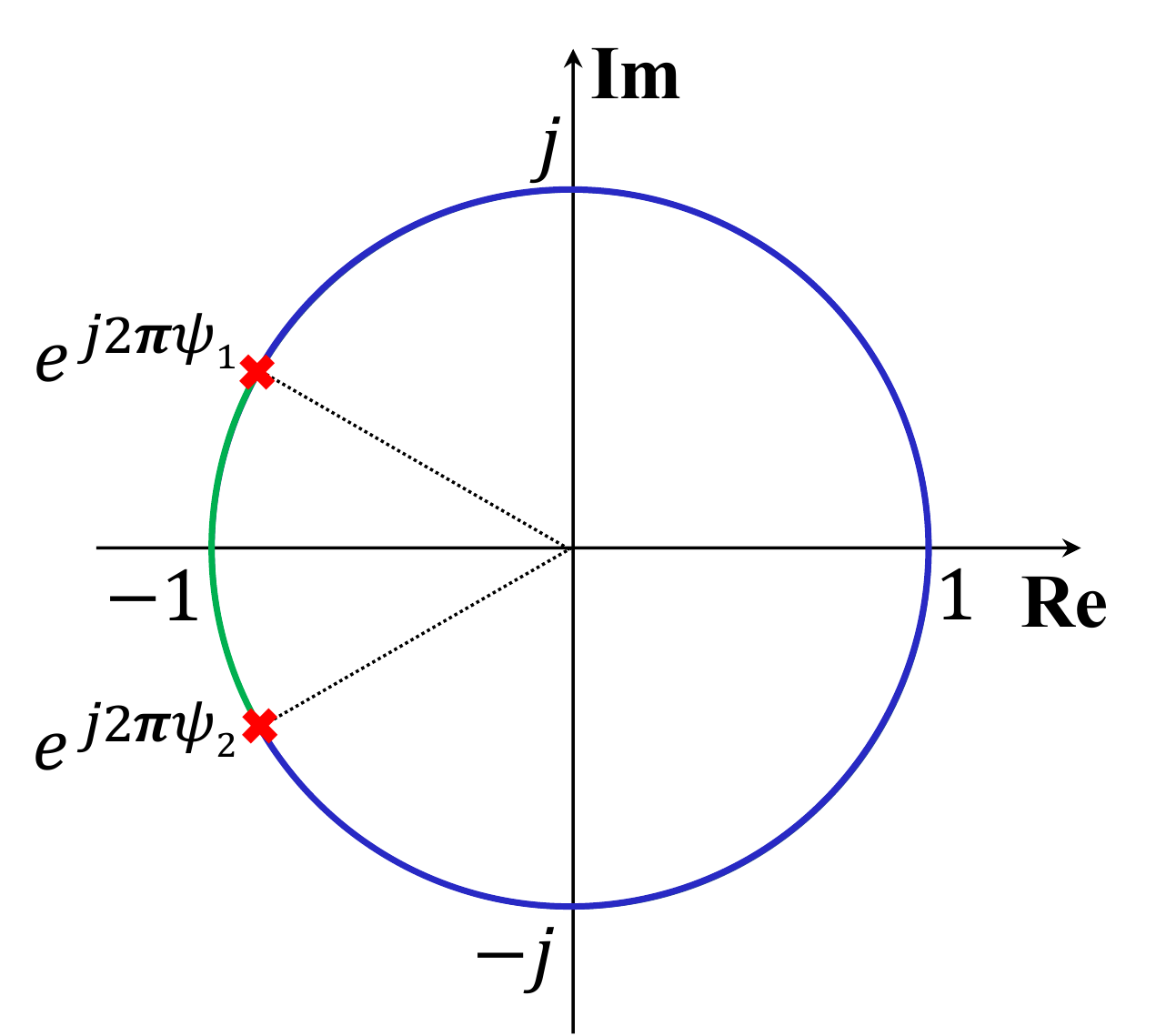} 
} 
\caption{Illustration of the wrapped least squares concept using complex values on the unit circle.} 
\label{fig:cycle} 
\end{figure}

A geometric interpretation of the WLS concept is given in Fig.~\ref{fig:cycle}. The absolute difference between a phase value $\phi_1$, which is an unwrapped version of an element of $\mathbf{\Psi}$, and the corresponding value in $\mathbf{H} \mathbf{R} \mathbf{X}_b$ ($\phi_2$), can be represented by the distance between two points $e^{j2\pi\phi_1}$ and $e^{j2\pi\phi_2}$ on the unit circle. Since $e^{j2\pi\psi_i} = e^{j 2 \pi \phi_i}$, where $\psi_i$ is the wrapped version of $\phi_i$, the absolute difference between $\phi_1$ and $\phi_2$ (on the unit circle) can be measured as the distance between the complex values $e^{j2\pi\psi_1}$ and $e^{j2\pi\psi_2}$ on the unit circle, i.e.,   
\begin{equation}
\left |{\operatorname{wrap}( {\psi_1 - \psi_2} )}\right | = \left |{\operatorname{wrap}( {\phi_1 - \phi_2} )}\right |.
\label{eq:unit cirle distance}
\end{equation}
The caveat here is that representing phase values on the unit circle automatically removes the effect of their integer parts. The distance of interest is always measured over the shorter (green) arc, which corresponds to respecting the half-cycle residual constraint \eref{res:eq2}. Given that only the fractional part of phase data is available, we can still measure the distance (on the unit circle) between the computed phase and the observed unwrapped phase (the observed wrapped phase plus the true integer value) using the fraction parts of the two phases, as emphasized by \eqref{eq:unit cirle distance}.

Based on the above discussion, estimating unknowns using phase measurements can be achieved by matching the fractional part of the phase observations and those of the corresponding predictions. If we consider a scalar phase observation $\psi_1$ and the prediction $\psi_2 = \psi (x)$ for an unknown parameter $x$, the estimation of $x$ can be carried out using
\begin{equation}
\mathop {\min }\limits_{{x} }\left |{\operatorname{wrap}( {\psi_1 - \psi_2} )}\right |.
\end{equation}
Extending the results to the vector/matrix case while considering the available set of constraints yields the proposed C-WLS optimization in \eref{eq3}.  

As compared with the C-ILS method, the C-WLS model maintains the integer constraint on the carrier-phase ambiguities (implicitly) and adds an additional constraint to limit the range of residual phase errors. Note that the C-WLS model keeps only the integer ambiguities satisfying the residual constraint, instead of the entire integer space. The effect of the residual restriction on the rotation matrix estimations can be demonstrated by comparing the solutions of (\ref{eq:eq23}) and (\ref{eq3}), which is the subject of the following lemma. 

\begin{lemma} \label{lem3}
A sufficient and necessary condition for (\ref{eq:eq23}) and (\ref{eq3}) to have the same global optimum regarding ${\mathbf{R}}$ is that the residual phase errors corresponding to the global minimum of (\ref{eq:eq23}) satisfy
\begin{equation}
\left| \xi_{ij}\!\left(\! {\mathbf{\hat R}}, \hat N_{ij} \!\right) \right|\! \leqslant \! \frac{1}{2}, \text{ } i \!=\! 1, 2, \cdots, \mathcal{S};\text{ } j \!=\! 1, 2, \cdots, \mathcal{A},
\label{eq:lem3}
\end{equation}
where $\mathbf{\hat R}$ and ${\mathbf{\hat N}}$ are the globally optimal solutions of (\ref{eq:eq23}), and $\hat N_{ij}$ is the entry in the $i$-th row and $j$-th column of ${\mathbf{\hat N}}$.
\end{lemma}
\begin{proof}
Suppose that $\mathbf{\hat R}$ and ${\mathbf{\hat N}}$ minimize (\ref{eq:eq23}), and they satisfy \eref{eq:lem3}. We can calculate the integer matrix based on $\mathbf{\hat R}$ using
\begin{equation}
{\mathbf{\tilde N}} ( {\mathbf{\hat R}} ) = \operatorname{round}( {{\bm{\Psi }} - {\mathbf{H}} {\mathbf{\hat  R}} {\mathbf{X}_b}} ).
\label{eq:lem3eq1}
\end{equation}
When the residual phase errors satisfy
\begin{equation}
\left| \xi_{ij}\!\left(\! {\mathbf{\hat R}}, \hat N_{ij} \!\right) \right|\! < \! \frac{1}{2}, \text{ } i \!=\! 1, 2, \cdots, \mathcal{S};\text{ } j \!=\! 1, 2, \cdots, \mathcal{A},
\label{eq:lem3eq2}
\end{equation}
we obtain ${\mathbf{\hat N}} = {\mathbf{\tilde N}} ( {\mathbf{\hat R}} )$. If there is some residual error $\left| \xi_{ij}\!\left(\! {\mathbf{\hat R}}, \hat N_{ij} \!\right) \right|\! =\! 1/2$, ${\mathbf{\hat N}}$ will have multiple possible values, and ${\mathbf{\tilde N}} ( {\mathbf{\hat R}} )$ is one of the solutions. Anyway, $\mathbf{\hat R}$ and ${\mathbf{\tilde N}} ( {\mathbf{\hat R}} )$ will minimize (\ref{eq:eq23}). Hence, $\mathbf{\hat R}$ will be the global optimum of \eref{eq:eq26}, which is equivalent to (\ref{eq3}).

In contrast, we assume that $\mathbf{\hat R}$ and ${\mathbf{\hat N}}$ minimize (\ref{eq:eq23}), and $\mathbf{\hat R}$ minimizes (\ref{eq3}). Given the fact that \eref{eq:eq26} and \eref{eq3} are equivalent, $\mathbf{\hat R}$ will also minimize \eref{eq:eq26}. Then, we have ${\mathbf{\hat N}} = {\mathbf{\tilde N}} ( {\mathbf{\hat R}} )$ or ${\mathbf{\tilde N}} ( {\mathbf{\hat R}} )$ is one of the possible values of ${\mathbf{\hat N}}$. Otherwise, $\mathbf{\hat R}$ cannot minimize $(\ref{eq:eq23})$ anymore. Therefore, $\mathbf{\hat R}$ and ${\mathbf{\hat N}}$ will satisfy \eref{eq:lem3}.
\end{proof}

The structure of the weight matrix ${\mathbf{Q}}_{\bm{\Psi }}$ is related to whether (\ref{eq:lem3}) holds or not. Likewise, the relationship between (\ref{eq:eq23}) and (\ref{eq3}) might differ for a various weight matrix ${\mathbf{Q}}_{\bm{\Psi }}$. For double-difference observations, ${{\mathbf{Q}}_{{\bm{\Psi }}}}$ is not diagonal. With common clock technology, multiple antennas can use a synchronized clock such that the single-difference model is applicable in attitude determination. In that case, ${{\mathbf{Q}}_{{\bm{\Psi }}}}$ is a diagonal matrix.
We obtain Lemma~\ref{lem2} as a special case to establish (\ref{eq:lem3}).



\begin{lemma} \label{lem2}
A sufficient condition for (\ref{eq:lem3}) to hold, and hence for (\ref{eq:eq23}) and (\ref{eq3}) to have the same global optimum for ${\mathbf{R}}$, is that ${{\mathbf{Q}}_{{\bm{\Psi }}}}$ is a diagonal matrix.
\end{lemma}
\begin{proof}
In (\ref{eq:eq23}), for $\forall {\mathbf{R}} \in {\mathbb{O}^{3 \times q}}$, we can figure out the optimal integer ambiguities by minimizing
\begin{equation}
F = \mathop {\min }\limits_{{\mathbf{N}} \in {\mathbb{Z}^{\mathcal{S} \times\mathcal{A}}}} \left\| \operatorname{vec} \! \left[\bm{\xi} ( {\mathbf{R}}, {\mathbf{N}} ) \right] \right\|_{{\mathbf{Q}}^{-1}_{{\bm{\Psi }}}}^2.
\end{equation}
The diagonal matrix ${{\mathbf{Q}}_{{\bm{\Psi }}}}$ results in the following expression
\begin{equation}
\begin{aligned}
F \!= \!&\mathop {\min }\limits_{N_{11} \in {\mathbb{Z}},\cdots,N_{_{\mathcal{S}\mathcal{A}}} \in {\mathbb{Z}}} \!\left(\! \sum\limits_{i = 1}^{\mathcal{S}} \sum\limits_{j = 1}^{\mathcal{A}} \frac{[\xi_{ij}( {\mathbf{R}}, N_{ij} )]^2}{\sigma _{\Psi ,kk}^2} \!\right)\!\\
\!=\!&\sum\limits_{i = 1}^{\mathcal{S}} \sum\limits_{j = 1}^{\mathcal{A}} \!\left( \!\mathop {\min }\limits_{N_{ij} \in {\mathbb{Z}}} \frac{[\xi_{ij}\!\left( {\mathbf{R}}, N_{ij} \right)]^2}{\sigma _{\Psi ,kk}^2} \!\right)\!, \text{ }k = (j\!-\!1)\mathcal{A}+i,
\end{aligned}
\end{equation}
where $\sigma _{\Psi ,kk}$ is the $k$-th diagonal entry of ${{\mathbf{Q}}_{{\bm{\Psi }}}}$. 


Since the sum entries are uncorrelated, we obtain the optimal by minimizing each entry independently. It is apparent that ${N_{ij} \in {\mathbb{Z}}}$ satisfying $|\xi_{ij}( {\mathbf{R}}, N_{ij} )| \leqslant 1/2$ will minimize $[\xi_{ij}( {\mathbf{R}}, N_{ij} )]^2$. Therefore, $\mathbf{\hat R}$ and ${\mathbf{\hat N}}$, the global optimums of (\ref{eq:eq23}), always satisfy (\ref{eq:lem3}) when ${{\mathbf{Q}}_{{\bm{\Psi }}}}$ is a diagonal matrix. According to Lemma~\ref{lem3}, we conclude that (\ref{eq:eq23}) and (\ref{eq3}) have the same global optimum for ${\mathbf{R}}$.
\end{proof}

To precisely explain the effect of the residual constraint, the results based on the unambiguous phase can be used as a benchmark to appreciate the difference between (\ref{eq:eq23}) and (\ref{eq3}). Assuming $\bm{\Phi}$ is the unambiguous double-different phase, we define ${\mathbb{S}}$, ${\mathbb{S}}_1$ and ${\mathbb{S}}_2$ as
\[{\mathbb{S}} = {\mathbb{O}^{3 \times q}},\]
\begin{equation*}
\resizebox{1\hsize}{!}{${\mathbb{S}}_1 \!\! =\!\! \{\mathbf{R}| {\mathbf{R}}\!\! \in \!\!{\mathbb{O}^{3 \times q}}\!, \bm{\xi} \!\!=\!\! {\bm{\Phi}}\!\!-\!\!{\mathbf{H}}{\mathbf{R}}{\mathbf{X}_b}, \left| \xi_{ij} \! \right| \!\!\leqslant\!\! \frac{1}{2}, i\!\! = \!\!1,\! \cdots \!, \mathcal{S};
j \!\!= \!\!1,\! \cdots \!, \mathcal{A}\}$},
\end{equation*}
\[{\mathbb{S}}_2={\mathbb{S}}-{\mathbb{S}}_1.\]
An \emph{oracle} estimator (that has access to the correct integer ambiguities) can be formulated as
\begin{equation}
  {{{\mathbf{R}}}_{_{\text{or}}} } = \mathop {\arg\min }\limits_{{\mathbf{R}} \in {\mathbb{S}_1}} \left\| \operatorname{vec} \! \left({\bm{\Phi} - {\mathbf{H}}{\mathbf{R}}{\mathbf{X}_b}} \right) \right\|_{{\mathbf{Q}}^{-1}_{{\bm{\Psi }}}}^2.
  \label{eq:eq3}
\end{equation}
Contrast this with the C-WLS estimator
\begin{equation}
  {{{\mathbf{R}}}_{_\text{CWLS}} } = \mathop {\arg\min }\limits_{{\mathbf{R}} \in {\mathbb{S}}} \left\| {{\operatorname{wrap}\!\left[\operatorname{vec} \! \left({\bm{\Psi}} - {\mathbf{H}}{\mathbf{R}}{\mathbf{X}_b}\right)\right]}} \right\|_{{\mathbf{Q}}^{-1}_{{\bm{\Psi }}}}^2,
  \label{eq:eq4}
\end{equation}
and the C-ILS estimator
\begin{equation}
  {{{\mathbf{R}}}_{_\text{CILS}} }, {{{\mathbf{N}}}_{_\text{CILS}} } \! = \!\mathop {\arg\min }\limits_{{\mathbf{R}} \in {\mathbb{S}},{\mathbf{N}} \in {\mathbb{Z}^{\mathcal{S} \times\mathcal{A}}}} \!\!\left\| \operatorname{vec} \! \left({
  {{\bm{\Psi }} \!-\! {\mathbf{H}}{\mathbf{R}}{\mathbf{X}_b}\! -\! {\mathbf{N}} }}\right) \right\|_{{\mathbf{Q}}^{-1}_{{\bm{\Psi }}}}^2.
  \label{eq:eq5}
\end{equation}
Corollary \ref{cor1}, Corollary \ref{cor2}, and Corollary \ref{cor3} summarize the relationships between these three estimators.

\begin{corollary} \label{cor1} If the C-ILS estimator (\ref{eq:eq5}) and the oracle estimator (\ref{eq:eq3}) have the same global optimum for ${\mathbf{R}}$, the C-WLS estimator (\ref{eq:eq4}) will also have the identical global optimum as the oracle estimator (\ref{eq:eq3}), i.e.,
\[{{{\mathbf{R}}}_{_\text{CILS}}} = {{\mathbf{R}}}_{_\text{or}}  \Rightarrow {{{\mathbf{R}}}_{_\text{CWLS}}} = {{\mathbf{R}}}_{_\text{or}}.\]
\end{corollary}
\begin{proof}
We can readily see that
\begin{equation}
\begin{aligned}
&\mathop {\min }\limits_{{\mathbf{R}} \in \mathbb{S}_1} \left\| {\operatorname{wrap}\left[\operatorname{vec} \! \left( {{\bm{\Psi}}  - {\mathbf{H}} {\mathbf{R}} {\mathbf{X}_b}} \right)\right]} \right\|_{{\mathbf{Q}}^{-1}_{{\bm{\Psi }}}}^2 \\
= &\mathop {\min }\limits_{{\mathbf{R}} \in \mathbb{S}_1} \|\operatorname{vec} \! \left( {{\bm{\Phi}}  - {\mathbf{H}} {\mathbf{R}} {\mathbf{X}_b}} \right)\|_{{\mathbf{Q}}^{-1}_{{\bm{\Psi }}}}^2,
\label{eq:eq6}
\end{aligned}
\end{equation}
and
\begin{equation}
\begin{aligned}
&\mathop {\min }\limits_{{\mathbf{R}} \in \mathbb{S}_2} \left\| {\operatorname{wrap}\left[\operatorname{vec} \! \left( {{\bm{\Psi}}  - {\mathbf{H}} {\mathbf{R}} {\mathbf{X}_b}}  \right)\right]} \right\|_{{\mathbf{Q}}^{-1}_{{\bm{\Psi }}}}^2 \\
 \geqslant & \mathop {\min }\limits_{{\mathbf{R}} \in \mathbb{S}_2,{\mathbf{N}} \in {\mathbb{Z}^{\mathcal{S} \times\mathcal{A}}}}  \left\| \operatorname{vec} \! \left(
  {{\bm{\Psi}} - {\mathbf{H}}{\mathbf{R}}{\mathbf{X}_b} - {\mathbf{N}} } \right)\right\|_{{\mathbf{Q}}^{-1}_{{\bm{\Psi }}}}^2.
\end{aligned}
\label{eq:eq7}
\end{equation}
If (\ref{eq:eq5}) converges to the oracle estimator, i.e., ${{{{\mathbf{R}}} _{_\text{CILS}} } \in \mathbb{S}_1}$ and ${{{\mathbf{R}}}_{_\text{CILS}}} = {{\mathbf{R}}}_{_\text{or}}$, we obtain
\begin{equation}
\begin{aligned}
&\mathop {\min }\limits_{{\mathbf{R}} \in \mathbb{S}_1,{\mathbf{N}} \in {\mathbb{Z}^{\mathcal{S} \times\mathcal{A}}}} \left\|\operatorname{vec} \! \left(
  {{\bm{\Psi }} - {\mathbf{H}}{\mathbf{R}}{\mathbf{X}_b} - {\mathbf{N}} }
 \right) \right\|_{{\mathbf{Q}}^{-1}_{{\bm{\Psi }}}}^2 \\
 < & \mathop {\min }\limits_{{\mathbf{R}} \in \mathbb{S}_2,{\mathbf{N}} \in {\mathbb{Z}^{\mathcal{S} \times\mathcal{A}}}} \left\| \operatorname{vec} \! \left( {{\bm{\Psi }} - {\mathbf{H}}{\mathbf{R}}{\mathbf{X}_b} -{\mathbf{N}} } \right)\right\|_{{\mathbf{Q}}^{-1}_{{\bm{\Psi }}}}^2,
 \end{aligned}
\label{eq:eq8}
\end{equation}
and 
\begin{equation}
\begin{aligned}
&\mathop {\min }\limits_{{\mathbf{R}} \in \mathbb{S}_1,{\mathbf{N}} \in {\mathbb{Z}^{\mathcal{S} \times\mathcal{A}}}} \! \left\| \operatorname{vec} \! \left( {{\bm{\Psi }} - {\mathbf{H}}{\mathbf{R}}{\mathbf{X}_b} - {\mathbf{N}} } \right) \right\|_{{\mathbf{Q}}^{-1}_{{\bm{\Psi }}}}^2 \\
 = &\mathop {\min }\limits_{{\mathbf{R}} \in \mathbb{S}_1} \! \| \operatorname{vec} \! \left({{\bm{\Phi}}  - {\mathbf{H}} {\mathbf{R}} {\mathbf{X}_b}} \right) \|_{{\mathbf{Q}}^{-1}_{{\bm{\Psi }}}}^2.
\end{aligned}
\label{eq:eq9}
\end{equation}
From \eref{eq:eq6}-\eref{eq:eq9}, it follows that
\begin{equation}
\begin{aligned}
&\mathop {\min }\limits_{{\mathbf{R}} \in \mathbb{S}_1} \! \left\| {\operatorname{wrap} \left[\operatorname{vec} \! \left( {{\bm{\Psi}}\!  -\! {\mathbf{H}} {\mathbf{R}} {\mathbf{X}_b}} \right)\right]}\! \right\|_{{\mathbf{Q}}^{-1}_{{\bm{\Psi }}}}^2 \\
< &\mathop {\min }\limits_{{\mathbf{R}} \in \mathbb{S}_2} \!\left\| {\operatorname{wrap}\left[\operatorname{vec} \! \left( {{\bm{\Psi}}  \!- \!{\mathbf{H}} {\mathbf{R}} {\mathbf{X}_b}} \right)\right]} \!\right\|_{{\mathbf{Q}}^{-1}_{{\bm{\Psi }}}}^2.
\end{aligned}
\end{equation}
Note that
\begin{equation}
\begin{aligned}
  &\mathop {\arg\min }\limits_{{\mathbf{R}} \in {\mathbb{S}_1}} \| \operatorname{vec} \! \left({\bm{\Phi} - {\mathbf{H}}{\mathbf{R}}{\mathbf{X}_b}} \right) \|_{{\mathbf{Q}}^{-1}_{{\bm{\Psi }}}}^2 \\
  = &\mathop {\arg\min }\limits_{{\mathbf{R}} \in {\mathbb{S}_1}} \left\| {{\operatorname{wrap}\left[\operatorname{vec} \! \left({\bm{\Psi}} - {\mathbf{H}}{\mathbf{R}}{\mathbf{X}_b} \right) \right]}} \right\|_{{\mathbf{Q}}^{-1}_{{\bm{\Psi }}}}^2,
  \end{aligned}
\end{equation}
then (\ref{eq:eq4}) will also have the same global optimum as \eref{eq:eq3}, i.e., ${{{{\mathbf{R}}}_{_\text{CWLS}}} \in \mathbb{S}_1}$ and ${{{\mathbf{R}}}_{_\text{CWLS}}} = {{{\mathbf{R}}}_{_\text{CILS}}} = {{\mathbf{R}}}_{_\text{or}}$. Hence, ${{{\mathbf{R}}}_{_\text{CILS}}} = {{\mathbf{R}}}_{_\text{or}}  \Rightarrow {{{\mathbf{R}}}_{_\text{CWLS}}} = {{\mathbf{R}}}_{_\text{or}}$, but not vice versa. 
\end{proof}

\begin{corollary} \label{cor2} 
If the C-WLS estimator \eref{eq:eq4} and the oracle estimator \eref{eq:eq3} have different global optimums, the global optimum (with respect to ${\mathbf{R}}$) of the C-ILS estimator \eref{eq:eq5} will also differ from that of the oracle estimator \eref{eq:eq3}. That is
\[{{{\mathbf{R}}}_{_\text{CWLS}}} \neq {{\mathbf{R}}}_{_\text{or}}  \Rightarrow {{{\mathbf{R}}}_{_\text{CILS}}} \neq {{\mathbf{R}}}_{_\text{or}}.\]
\end{corollary}
\begin{proof}
If ${{{\mathbf{R}}}_{_\text{CWLS}}} \neq {{\mathbf{R}}}_{_\text{or}}$, then ${{{{\mathbf{R}}}_{_\text{CWLS}}} \in \mathbb{S}_2}$, and
\begin{equation}
\begin{aligned}
&\mathop {\min }\limits_{{\mathbf{R}} \in \mathbb{S}_1} \! \left\| {\operatorname{wrap}\! \left[\operatorname{vec} \! \left( {{\bm{\Psi}}  \!- \!{\mathbf{H}} {\mathbf{R}} {\mathbf{X}_b}} \right)\right]} \!\right\|_{{\mathbf{Q}}^{-1}_{{\bm{\Psi }}}}^2 \\
> &\mathop {\min }\limits_{{\mathbf{R}} \in \mathbb{S}_2} \! \left\| {\operatorname{wrap}\!\left[\operatorname{vec} \! \left( {{\bm{\Psi}} \! - \! {\mathbf{H}} {\mathbf{R}} {\mathbf{X}_b}} \right)\right]} \!\right\|_{{\mathbf{Q}}^{-1}_{{\bm{\Psi }}}}^2.
\label{eq:eq10}
\end{aligned}
\end{equation}
According to (\ref{eq:eq7}) and (\ref{eq:eq10}), we have
\begin{equation}
\begin{aligned}
& \mathop {\min }\limits_{{\mathbf{R}} \in \mathbb{S},{\mathbf{N}} \in {\mathbb{Z}^{\mathcal{S} \times\mathcal{A}}}} \left\| \operatorname{vec} \! \left(
  {{\bm{\Psi }} - {\mathbf{H}}{\mathbf{R}}{\mathbf{X}_b} - {\mathbf{N}} } \right) \right\|_{{\mathbf{Q}}^{-1}_{{\bm{\Psi }}}}^2 \\
\leqslant & \mathop {\min }\limits_{{\mathbf{R}} \in \mathbb{S}_2,{\mathbf{N}} \in {\mathbb{Z}^{\mathcal{S} \times\mathcal{A}}}} \left\|\operatorname{vec} \! \left({{\bm{\Psi }} - {\mathbf{H}}{\mathbf{R}}{\mathbf{X}_b} - {\mathbf{N}} } \right) \right\|_{{\mathbf{Q}}^{-1}_{{\bm{\Psi }}}}^2\\
\leqslant &\mathop {\min }\limits_{{\mathbf{R}} \in \mathbb{S}_2} \left\| {\operatorname{wrap}\left[\operatorname{vec} \! \left( {{\bm{\Psi}}  - {\mathbf{H}} {\mathbf{R}} {\mathbf{X}_b}} \right) \right]} \right\|_{{\mathbf{Q}}^{-1}_{{\bm{\Psi }}}}^2 \\
< &\left\| {\operatorname{wrap}\left[\operatorname{vec} \! \left( {{\bm{\Psi}}  - {\mathbf{H}} {\mathbf{\hat R}_{_\text{or}}}} \right)\right]} \right\|_{{\mathbf{Q}}^{-1}_{{\bm{\Psi }}}}^2 \\
= &\left\| \operatorname{vec} \! \left({{\bm{\Psi }} - {\mathbf{H}} {\mathbf{\hat x}_{_\text{or}}} - {{{\mathbf{N}}}_{_\text{or}}} }\right)\right\|_{{\mathbf{Q}}^{-1}_{{\bm{\Psi }}}}^2,
\end{aligned}
\end{equation}
with
\begin{equation}
{{{\mathbf{N}}}_{_\text{or}}} = {\bm{\Psi }} - {\bm{\Phi }},
\end{equation}
then ${{{\mathbf{R}}}_{_\text{CILS}}} \neq {{\mathbf{R}}}_{_\text{or}}$ and ${{{\mathbf{N}}}_{_\text{CILS}}} \neq {{{\mathbf{N}}}_{_\text{or}}}$. As a result, ${{{\mathbf{R}}}_{_\text{CWLS}}} \neq {{\mathbf{R}}}_{_\text{or}}  \Rightarrow {{{\mathbf{R}}}_{_\text{CILS}}} \neq {{\mathbf{R}}}_{_\text{or}}$. 
\end{proof}

Using the results of Corollary \ref{cor1} and Corollary \ref{cor2},  the following corollary holds.

\begin{corollary} \label{cor3} Compared with the C-ILS estimator (\ref{eq:eq5}), the C-WLS estimator (\ref{eq:eq4}) is more likely to have the equivalent global optimum for ${\mathbf{R}}$ as the oracle estimator (\ref{eq:eq3}), i.e.,
\[P({{{\mathbf{R}}}_{_\text{CWLS}}} = {{\mathbf{R}}}_{_\text{or}}) \geqslant P({{{\mathbf{R}}}_{_\text{CILS}}} = {{\mathbf{R}}}_{_\text{or}}),\]
where $P(\cdot)$ denotes the probability.
\end{corollary}
\begin{proof}
Since
\[{{{\mathbf{R}}}_{_\text{CILS}}} = {{\mathbf{R}}}_{_\text{or}}  \Rightarrow {{{\mathbf{R}}}_{_\text{CWLS}}} = {{\mathbf{R}}}_{_\text{or}},\]
\[{{{\mathbf{R}}}_{_\text{CWLS}}} \neq {{\mathbf{R}}}_{_\text{or}}  \Rightarrow {{{\mathbf{R}}}_{_\text{CILS}}} \neq {{\mathbf{R}}}_{_\text{or}},\]
it can be readily seen that
$P({{{\mathbf{R}}}_{_\text{CWLS}}} \!\!=\!\! {{\mathbf{R}}}_{_\text{or}}) \!\!\geqslant\!\! P({{{\mathbf{R}}}_{_\text{CILS}}} \!\! =\!\! {{\mathbf{R}}}_{_\text{or}}).$
\end{proof}

When only carrier-phase measurements are utilized to jointly estimate the ambiguities and attitude, there are more unknowns than the number of equations. Corollary \ref{cor3} illuminates the advantage of the C-WLS approach compared to the C-ILS method. However, the pseudo-range observations are still required to improve the solutions.

\begin{lemma} \label{lem5}
Including both pseudo-range and carrier-phase data, the proposed C-WLS method can be formulated as
\begin{equation}
\mathop {\min }\limits_{{\mathbf{R}} \in {\mathbb{S}}} \left\| {\left[ \!\! {\begin{array}{*{20}{c}}
  {\operatorname{wrap} \! \left[\operatorname{vec} \! \left( {{\bm{\Psi }} - {\mathbf{H}} {\mathbf{R}} {\mathbf{X}_b}} \right) \right]} \\ 
  \operatorname{vec} \! \left({{\bf{P}}  - {\mathbf{H}} {\mathbf{R}} {\mathbf{X}_b}} \right)
\end{array}} \!\! \right]} \right\|_{{\mathbf{Q}^{-1}}}^2. 
\label{eq:eq25}
\end{equation}
\end{lemma}
It is readily seen that (\ref{eq:eq25}) will achieve the same global minimum and optimal solution of $\mathbf{R}$ as
\begin{subequations}
\label{eq13}
\begin{align}
\mathop {\min }\limits_{{\mathbf{R}} \in {\mathbb{S}},{\mathbf{N}} \in {\mathbb{Z}^{\mathcal{S} \times\mathcal{A}}}} & \left\| \operatorname{vec} \! \left({{\mathbf{Y}}  - {{\mathbf{A}} \mathbf{R}\mathbf{X}_b} - {\mathbf{B}} {\mathbf{N}}} \right) \right\|_{{\mathbf{Q}_{\mathbf{Y}}^{-1}}}^2 \\
\text{s.t.} & -\frac{1}{2} < \xi_{ij}( {\mathbf{R}}, N_{ij} ) \leqslant \frac{1}{2},
\end{align}
\end{subequations}
\[i = 1, 2, \cdots, \mathcal{S};\quad j= 1, 2, \cdots, \mathcal{A}.\]
To compare \eref{eq:eq25} with \eref{eq:cls}, we assume that
\begin{equation}
  {{{\mathbf{\tilde R}}}_{_\text{or}} } = \mathop {\arg\min }\limits_{{\mathbf{R}} \in {\mathbb{S}_1}} \left\| \operatorname{vec} \! \left({\left[ \!\!{\begin{array}{*{20}{c}}
  \bm{\Phi}\\ 
  \bf{P} 
\end{array}}\!\! \right]}  - {{\left[ \!\!{\begin{array}{*{20}{c}}
  \mathbf{H}\\ 
  \mathbf{H} 
\end{array}} \!\! \right]} \mathbf{R}\mathbf{X}_b} \right)\right\|_{{\mathbf{Q}^{-1}}}^2,
\end{equation}
\begin{equation}
  {{{\mathbf{\tilde R}}}_{_\text{CWLS}} } = \mathop {\arg\min }\limits_{{\mathbf{R}} \in {\mathbb{S}}} \left\| {\left[ \!\!{\begin{array}{*{20}{c}}
  {\operatorname{wrap} \! \left[\operatorname{vec} \! \left( {{\bm{\Psi }} - {\mathbf{H}} {\mathbf{R}} {\mathbf{X}_b}} \right)\right]} \\ 
  \operatorname{vec} \! \left({{\bf{P}}  - {\mathbf{H}} {\mathbf{R}} {\mathbf{X}_b}} \right)
\end{array}}\!\! \right]} \right\|_{{\mathbf{Q}^{-1}}}^2,
\end{equation}
\begin{equation}
  {{{\mathbf{\tilde R}}}_{_\text{CILS}} }, {{{\mathbf{\tilde N}}}_{_\text{CILS}} } \! =\! \mathop {\arg\min }\limits_{{\mathbf{R}} \in {\mathbb{S}},{\mathbf{N}} \in {\mathbb{Z}^{\mathcal{S} \times\mathcal{A}}}} \!\! \left\|\! \operatorname{vec} \! \left({{\mathbf{Y}}  \!-\! {{\mathbf{A}} \mathbf{R}\mathbf{X}_b} \!-\! {\mathbf{B}} \mathbf{N}} \right)\!  \right\|_{{\mathbf{Q}_{\mathbf{Y}}^{-1}}}^2. \!\!\!
\end{equation}
Given the structure of ${\mathbf{Q}}$, we have
\begin{equation}
\begin{aligned}
{{{\mathbf{\tilde R}}}_{_\text{or}} } \! = \! \mathop {\arg\min }\limits_{{\mathbf{R}} \in {\mathbb{S}_1}}   \{ &\left\| 
  \operatorname{vec} \! \left({{\bm{\Phi }} \! - \! {\mathbf{H}}{\mathbf{R}} {\mathbf{X}_b}} \right) \right\|_{{\mathbf{Q}^{-1}_{{\bm{\Psi }}}}}^2 \\
  +  &\left\|\operatorname{vec} \! \left(
  {{\bf{P}} \! - \!{\mathbf{H}}{\mathbf{R}} {\mathbf{X}_b}}\right) \right\|_{{\mathbf{Q}^{-1}_{{\bf{P}}}}}^2 \},
\label{eq:eq11}
\end{aligned}
\end{equation}
\begin{equation}
\begin{aligned}
{{{\mathbf{\tilde R}}}_{_\text{CWLS}} }\!\! = \! \mathop {\arg\min }\limits_{{\mathbf{R}} \in {\mathbb{S}}}  \{ &\left\| 
  {\operatorname{wrap}\!\left[\operatorname{vec} \! \left( {{\bm{\Psi }}\! - \!{\mathbf{H}} {\mathbf{R}} {\mathbf{X}_b}} \right)\right]}  \right\|_{{\mathbf{Q}^{-1}_{{\bm{\Psi }}}}}^2 \\  + & \left\| \operatorname{vec} \! \left(
  {{\bf{P}}\!  - \!{\mathbf{H}}{\mathbf{R}} {\mathbf{X}_b}}\right) \right\|_{{\mathbf{Q}^{-1}_{{\bf{P}}}}}^2 \} , 
  \label{eq:eq12}
\end{aligned}
\end{equation}
\begin{equation}
\begin{aligned}
{{{\mathbf{\tilde R}}}_{_\text{CILS}} }, {{{\mathbf{\tilde N}}}_{_\text{CILS}} }\! \! =\! \!\mathop {\arg\min }\limits_{{\mathbf{R}} \in {\mathbb{S}},{\mathbf{N}} \in {\mathbb{Z}^{\mathcal{S} \times\mathcal{A}}}}   \{&\left\| \!
  \operatorname{vec} \! \left({{\bm{\Psi }}\! - \!{\mathbf{H}} {\mathbf{R}} {\mathbf{X}_b}}\! - \!{{{\mathbf{N}}}}\right)  \!\right\|\! _{{\mathbf{Q}^{-1}_{{\bm{\Psi}}}}}^2 \\+ &\left\| \!\operatorname{vec} \! \left(
  {{\bf{P}}  \! - \! {\mathbf{H}}{\mathbf{R}} {\mathbf{X}_b}}\right) \!\right\| \!_{{\mathbf{Q}^{-1}_{{\bf{P}}}}}^2\}.
  \label{eq:eq13}
\end{aligned}
\end{equation}
Incorporating pseudo-range observations will introduce the same terms ($\left\|\operatorname{vec} \! \left({{\bf{P}}  - {\mathbf{H}} {\mathbf{R}}{\mathbf{X}_b}}\right) \right\|_{{\mathbf{Q}^{-1}_{{\bf{P}}}}}^2$) to \eref{eq:eq11}-\eref{eq:eq13}. 
So all the results in Corollary \ref{cor1}, Corollary \ref{cor2}, and Corollary \ref{cor3} are still true. As a consequence, we establish the following corollary.
\begin{corollary} \label{cor4} 
For the estimators in (\ref{eq:eq11}), (\ref{eq:eq12}) and (\ref{eq:eq13}), the global optimums with respect to ${\mathbf{R}}$ satisfy the following relationships
\[{{{\mathbf{\tilde R}}}_{_\text{CILS}}} = {{\mathbf{\tilde R}}}_{_\text{or}}  \Rightarrow {{{\mathbf{\tilde R}}}_{_\text{CWLS}}} = {{\mathbf{\tilde R}}}_{_\text{or}},\]
\[{{{\mathbf{\tilde R}}}_{_\text{CWLS}}} \neq {{\mathbf{\tilde R}}}_{_\text{or}}  \Rightarrow {{{\mathbf{\tilde R}}}_{_\text{CILS}}} \neq {{\mathbf{\tilde R}}}_{_\text{or}},\] \[P({{{\mathbf{\tilde R}}}_{_\text{CWLS}}} = {{\mathbf{\tilde R}}}_{_\text{or}}) \geqslant P({{{\mathbf{\tilde R}}}_{_\text{CILS}}} = {{\mathbf{\tilde R}}}_{_\text{or}}).\]
\end{corollary}
Corollary~\ref{cor4} demonstrates the advantage of the C-WLS model. The proposed method imposes a constraint on the residual errors, which shrinks the solution space and remove spurious solutions. This results in an improvement over the C-ILS method. To give insight into all the results discussed in this section, we summarize the significance of various lemmas and corollaries in Remark~\ref{remk1}.

\begin{remark} \label{remk1}
The significance of various lemmas and corollaries is as follows.

Lemma~\ref{lem1} and Lemma~\ref{lem5} formulate the proposed C-WLS optimization model.
\begin{itemize}
\item Lemma~\ref{lem1} based on only carrier phase data.
\item Lemma~\ref{lem5} based on pseudo-range and carrier phase data.
\end{itemize}
Lemma~\ref{lem3} -- \ref{lem2} discuss the conditions for the C-WLS and the C-ILS optimization to have equivalent global optimum for ${\mathbf{R}}$ when using only carrier phase observations.
\begin{itemize}
\item Lemma~\ref{lem3} provides a sufficient and necessary condition.
\item Lemma~\ref{lem2} proves a diagonal ${\mathbf{Q}}_{\bm{\Psi }}$ as a sufficient condition.
\end{itemize}
Corollary~\ref{cor1} -- \ref{cor4} indicate the advantage of the C-WLS model compared with the C-ILS model; namely, the C-WLS estimator is more likely to have the same global optimum for ${\mathbf{R}}$ as the oracle estimator.
\begin{itemize}
\item Corollary~\ref{cor1} -- \ref{cor3} based on only carrier phase data. 
\item Corollary~\ref{cor4} based on pseudo-range and carrier phase data.
\end{itemize}
\end{remark}

\section{Implementation of the Proposed C-WLS Method}
\label{sec:solution}
In the C-WLS problem (\ref{eq:eq25}), the rotation matrix ${\mathbf{R}}$ is the only unknown. For $q = 3$, ${\mathbf{R}}$ includes 9 elements that are uniquely determined by 3 independent variables. The Euler angles are usually used to define ${\mathbf{R}}$ and represent the orientation of the rigid body with respect to the three axes in reference coordinate system. The rotation matrix ${\mathbf{R}}$ can be expressed in the following form:
\begin{equation*}
\resizebox{.99\hsize}{!}{
${\mathbf{R}} \! = \! \left[ \!\!\! {\begin{array}{*{20}{c}}
  {\cos \alpha \cos \beta } &{\cos \alpha \sin \beta \sin \gamma \! - \! \sin \alpha \cos \gamma } &{\cos \alpha \sin \beta \cos \gamma \! + \! \sin \alpha \sin \gamma } \\ 
  {\sin \alpha \cos \beta } &{\sin \alpha \sin \beta \sin \gamma \! + \! \cos \alpha \cos \gamma } &{\sin \alpha \sin \beta \cos \gamma  \!-\! \cos \alpha \sin \gamma } \\ 
  { - \! \sin \beta } &{\cos \beta \sin \gamma } &{\cos \beta \cos \gamma } 
\end{array}} \!\!\! \right]$},
\end{equation*}
where $\alpha$, $\beta$ and $\gamma$ are the yaw, pitch, and roll angles, respectively. An optimal of (\ref{eq:eq25}) can be found by searching over the range of the Euler angles. However, such a 3-D search leads to a high computational complexity. Given the prior knowledge of the antenna-array configuration and the integer characteristics of phase ambiguities, we can efficiently reduce the search space to a few possible estimates of ${\mathbf{R}}$ (instead of searching through the whole range of the Euler angles). In the following subsections, we will develop the search strategies for the single-baseline and multi-baseline cases separately.

\subsection{Single-baseline Attitude Determination}
In a single-baseline set-up, ${\mathbf{R}} \in {\mathbb{R}^{3}}$ is a unit column vector, uniquely determined by yaw and pitch angles. To facilitate the discussion of the multi-baseline case (presented later), we replace ${\mathbf{R}}$ with ${\mathbf{r}_a}$, which will later be used to indicate the $a$-th (single) baseline in a multi-baseline configuration. The vector ${\mathbf{r}_a}$ represents the unit direction vector of a single baseline. 
Based on (\ref{eq:eq25}), we can write
\begin{equation}
\mathop {\min }\limits_{{\mathbf{r}_a} \in {\mathbb{R}^{3}}, \left\| {\mathbf{r}_a} \right\|_2 = 1} \left\| {\left[ \!\!\! {\begin{array}{*{20}{c}}
  {\operatorname{wrap}( {{\bm{\psi}}_a - d_a{\mathbf{H}} {\mathbf{r}_a}} )} \\ 
  {{\bm{\rho}}_a  - d_a{\mathbf{H}} {\mathbf{r}_a}} 
\end{array}} \!\!\! \right]} \right\|_{{\mathbf{Q}_a^{-1}}}^2,
\label{eq:singlecwls}
\end{equation}
where $d_a$ is the baseline length, and ${\mathbf{Q}_a}$ is the covariance matrix. The baseline vector in the reference frame is given by
\begin{equation}
{\mathbf{x}_a} = d_a {\mathbf{r}_a}.
\label{eq:base_vec}
\end{equation}

According to (\ref{dd:eq2}), (\ref{eq:H}) and (\ref{eq:base_vec}), we have
\begin{equation}
 \psi_{a0}^{s0} -  {{N}} _{a0}^{s0} = d_a {{\bf{\tilde h}}_s}{\mathbf{r}_a} + \eta _{a0}^{s0}.
\end{equation}
As shown in Fig.~\ref{fig:intersection1}, we draw the vector ${\mathbf{r}_a}$ and ${{\bf{\tilde h}}_s}$ starting from the origin. ${\mathbf{r}_a}$ is the unknown unit vector which makes an angle $\theta_s$ with ${{\bf{\tilde h}}_s}$. Hence, the candidate terminal points (coordinates) for ${\mathbf{r}_a}$ represent a circle on the sphere centered around the radial line of ${{\bf{\tilde h}}_s}$. If we ignore phase noise, the angle $\theta_s$ satisfies 
\begin{equation}
 \cos \theta_s = \frac{\psi_{a0}^{s0} -  {{N}} _{a0}^{s0}}{d_a \| {\bf{\tilde h}}_s \|_2}.
 \label{eq:cos}
\end{equation}
For the $s$-th element of ${\bm{\psi}}_a$, the possible integer ambiguities should be confined to
\begin{equation}
 -{d_a \| {\bf{\tilde h}}_s \|_2} \leqslant \psi_{a0}^{s0} -  {{N}} _{a0}^{s0} \leqslant {d_a \| {\bf{\tilde h}}_s \|_2}.
\end{equation}
Considering all the possible ${{N}} _{a0}^{s0}$, the candidate terminal points for ${\mathbf{r}_a}$ will be on a set of parallel circles on the spherical surface.
\begin{figure}[tbp]
\centering
\includegraphics[width=0.30\textwidth]{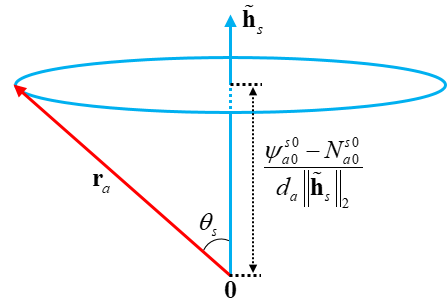}
\caption{The relationship between ${\mathbf{r}_a}$ and ${{\bf{\tilde h}}_s}$.}
\label{fig:intersection1}
\end{figure}
\begin{figure}[tbp]
\centering
\includegraphics[width=0.37\textwidth]{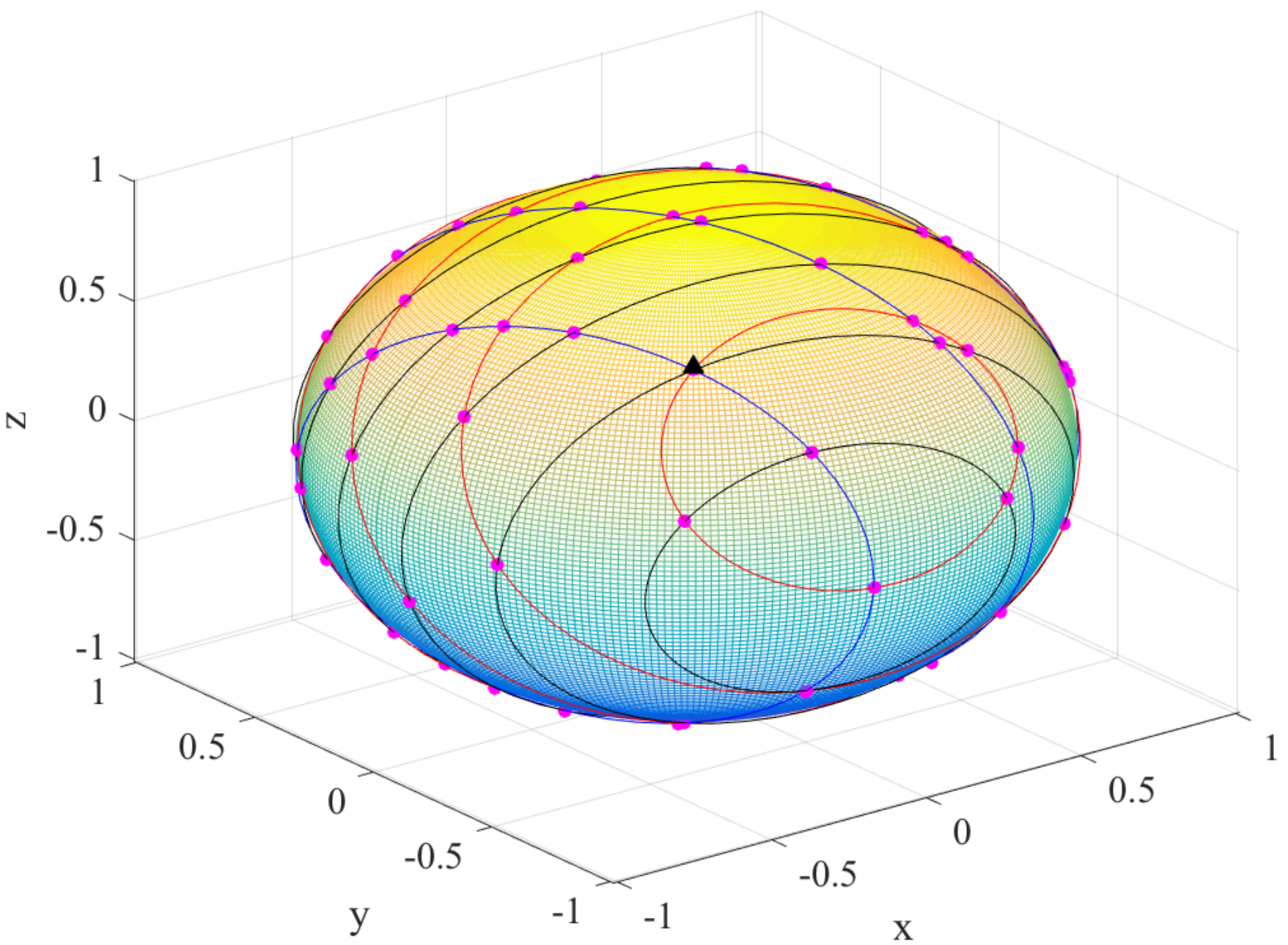}
\caption{Distribution of potential candidate points for ${\mathbf{r}_a}$.}
\label{fig:intersection2}
\end{figure}

Fig.~\ref{fig:intersection2} shows an example of the distribution of possible candidate points for ${\mathbf{r}_a}$ with four satellites under consideration (circles of the same color corresponding to one double-difference phase observation). The intersections between the circles of different colors indicate possible estimates of ${\mathbf{r}_a}$. The intersection points between two circles exist if the following system of equations has a solution:
\begin{equation}
\left\{\begin{matrix}
\begin{aligned}
&\psi_{a0}^{s0} - {{N}} _{a0}^{s0} = d_a {{\bf{\tilde h}}_s}{\mathbf{r}_a}, \\ 
&\psi_{a0}^{m0} - {{N}} _{a0}^{m0} = d_a {{\bf{\tilde h}}_m}{\mathbf{r}_a}.
\end{aligned}
\end{matrix}\right.
\label{eq:sys1}
\end{equation}
In the noise-free case, there is one combination of the integers
\begin{equation}
{\bf{N }}_a = {\left[ \!\!\! {\begin{array}{*{20}{c}}
  {{N }_{a0}^{10}}& \cdots &{{N }_{a0}^{\mathcal{S} 0}}
\end{array}} \!\!\! \right]^{\text{T}}},
\label{eq:int}
\end{equation}
whose corresponding circles have a joint intersection point,
that is, \eref{eq:int} satisfies
\begin{equation}
\left\{\begin{matrix}
\begin{aligned}
\psi_{a0}^{10} - {{N}} _{a0}^{10} &= d_a {{\bf{\tilde h}}_1}{\mathbf{r}_a}, 
\\ \psi_{a0}^{20} - {{N}} _{a0}^{20} &= d_a {{\bf{\tilde h}}_2}{\mathbf{r}_a}, \\
&\vdots\\
\psi_{a0}^{\mathcal{S} 0} - {{N}} _{a0}^{\mathcal{S} 0} &= d_a {{\bf{\tilde h}}_{_\mathcal{S}}}{\mathbf{r}_a}. 
\end{aligned}
\end{matrix}\right.
\label{eq:sys2}
\end{equation}
Due to the measurement noise, we may have only pairwise circle intersect,
that is, a solution of (\ref{eq:sys2}) may not exist. However, we can calculate the intersection points for all the circle pairs of different observations. The baseline vector can be determined by evaluating (\ref{eq:singlecwls}) (only) at these intersections.

Instead of solving the system of equations (\ref{eq:sys1}), we can achieve the intersections by applying the geometry relationships between two circles on the sphere. To calculate the intersection points, we need to define some intermediate vectors, such as the normal vector ${{\bf{v}}_1}$ of the plane determined by ${{\bf{\tilde h}}_s}$ and ${{\bf{\tilde h}}_m}$, and the normal vector ${{\bf{v}}_2}$ of the plane on which ${{\bf{v}}_1}$ and ${{\bf{\tilde h}}_m}$ lie. ${{\bf{v}}_1}$ and ${{\bf{v}}_2}$ are given by
\begin{equation}
{{\bf{v}}_1} = \frac{{{\bf{\tilde h}}_s} \times {\bf{\tilde h}}_m}{\| {{\bf{\tilde h}}_s} \times {\bf{\tilde h}}_m \|_2}, \text{ }
{{\bf{v}}_2} = \frac{{{\bf{v}}_1} \times {\bf{\tilde h}}_m}{\| {{\bf{v}}_1} \times {\bf{\tilde h}}_m \|_2}.
\label{eq:int1}
\end{equation}
In Fig.~\ref{fig:intersection3}, the centers of the circle $s$ and circle $m$ are
\begin{equation}
 \mathbf{c}_s = \frac{\cos \theta_s{{\bf{\tilde h}}_s}}{\| {\bf{\tilde h}}_s \|_2}, \text{ }
 \mathbf{c}_m = \frac{\cos \theta_m{{\bf{\tilde h}}_m}}{\| {\bf{\tilde h}}_m \|_2},
 \label{eq:int2}
\end{equation}
where the cosines can be computed based on (\ref{eq:cos}).

To check whether the two circles intersect, we define the peak points (the closest or farthest points from the other circle) of circle $m$ with respect to circle $s$ as follows
\begin{equation}
{{\bf{p}}_1} = {{\bf{c}}_m} - \sin \theta_m {{\bf{v}}_2}, \text{ }
{{\bf{p}}_2} = {{\bf{c}}_m} + \sin \theta_m {{\bf{v}}_2}.
\label{eq:int3}
\end{equation}
The directed distances (positive or negative) between the peak points and circle $s$ plane are given by
\begin{equation}
\Delta_1= \left({{\bf{p}}_1} - {{\bf{c}}_s} \right)^{\text{T}}\frac{{\bf{\tilde h}}_s}{\| {\bf{\tilde h}}_s \|_2}, \text{ }
\Delta_2 = \left({{\bf{p}}_2} - {{\bf{c}}_s} \right)^{\text{T}}\frac{{\bf{\tilde h}}_s}{\| {\bf{\tilde h}}_s \|_2}.
\label{eq:int4}
\end{equation}
If $\Delta_1\Delta_2 \leqslant 0$, circle $s$ will intersect with circle $m$. Otherwise, there is no intersection.
When $\Delta_1 = 0$, the two circles intersect at ${{\bf{p}}_1}$. Similarly, if $\Delta_2 = 0$, ${{\bf{p}}_2}$ is the intersection point. When $\Delta_1\Delta_2 < 0$, there are two intersection points between circle $s$ and circle $m$. A procedure to calculate these intersection points is as follows. The directed distance between ${{\bf{c}}_m}$ and circle $s$ plane is
\begin{equation}
\Delta_3 = \left({{\bf{c}}_m} - {{\bf{c}}_s} \right)^{\text{T}}\frac{{\bf{\tilde h}}_s}{\| {\bf{\tilde h}}_s \|_2}.
\label{eq:int5}
\end{equation}
The projected point of ${{\bf{c}}_m}$ on circle $s$ plane is expressed as
\begin{equation}
{{\bf{p}}_{{\bf{c}}_m}} = {{\bf{c}}_m} - \frac{\Delta_3 {{\bf{\tilde h}}_s}}{\| {\bf{\tilde h}}_s \|_2}.
\label{eq:int6}
\end{equation}
As shown in Fig.~\ref{fig:intersection3}, ${{\bf{\tilde h}}_s}$ intersects with the circle $m$ plane at
\begin{equation}
{{\bf{p}}_{{\bf{c}}_s}} = \frac{\cos \theta_{m}}{\cos \theta_{sm}} \frac{{\bf{\tilde h}}_s}{\| {\bf{\tilde h}}_s \|_2},
\label{eq:int7}
\end{equation}
where $\theta_{sm}$ is the angle between ${{\bf{\tilde h}}_s}$ and ${{\bf{\tilde h}}_m}$. The directed distance between ${{\bf{p}}_{{\bf{c}}_s}}$ and circle $s$ plane is given by
\begin{equation}
\Delta_4 = \left({{\bf{p}}_{{\bf{c}}_s}}- {{\bf{c}}_s} \right)^{\text{T}}\frac{{\bf{\tilde h}}_s}{\| {\bf{\tilde h}}_s \|_2}.
\label{eq:int8}
\end{equation}
Based on the basic proportionality theorem of similar triangles, we have
\begin{equation}
\frac{{\bf{p}}_c - {\bf{c}}_s}{{\bf{p}}_c - {{\bf{p}}_{{\bf{c}}_m}}} = \frac{\Delta_4}{\Delta_3},
\label{eq:int9}
\end{equation}
where ${\bf{p}}_c$ is the center point on the intersection line. Therefore, we have
\begin{equation}
{{\bf{p}}_c} = \frac{\Delta_3{{\bf{c}}_s}-\Delta_4{{\bf{p}}_{{\bf{c}}_m}}}{\Delta_3-\Delta_4}.
\end{equation}
The two intersection points can be represented as
\begin{equation}
{{\bf{q}}_1} = {{\bf{p}}_c} - \sqrt{1 - {\left\| {\bf{p}}_c \right\|}_2^2} {{\bf{v}}_1}, \text{ }
{{\bf{q}}_2} = {{\bf{p}}_c} + \sqrt{1 - {\left\| {\bf{p}}_c \right\|}_2^2} {{\bf{v}}_1}.
\label{eq:int10}
\end{equation}

We consider all the intersection points as candidates for being an estimate of ${\mathbf{r}_a}$. When there is no intersection, if $\Delta_1$ or $\Delta_2$ is close to zero, some points around the corresponding peak point should also be selected as the potential candidates as the lack of intersection might be due to noise. We define a threshold $\delta_\Delta$ for $\Delta_1$ and $\Delta_2$. If $|\Delta_1| < \delta_\Delta$, we can select ${{\bf{p}}_1}$ as a candidate; if $|\Delta_2| < \delta_\Delta$, we can choose ${{\bf{p}}_2}$ as a candidate. 
We can create a pool (${\mathbb{S}}_{\mathbf{r}_a}$) of all candidate points from the intersections of all possible circle pair combinations. Algorithm~\ref{alg0} summarizes the procedure to determine ${\mathbb{S}}_{\mathbf{r}_a}$. We keep the best $K$ estimations of ${\mathbf{r}_a}$ (coarse solutions) using
\begin{equation}
\mathop {\min }\limits_{{\mathbf{r}_a} \in {\mathbb{S}}_{\mathbf{r}_a}} \left\| {\left[ \!\!\! {\begin{array}{*{20}{c}}
  {\operatorname{wrap}( {{\bm{\psi}}_a - d_a{\mathbf{H}} {\mathbf{r}_a}} )} \\ 
  {{\bm{\rho}}_a  - d_a{\mathbf{H}} {\mathbf{r}_a}} 
\end{array}} \!\!\! \right]} \right\|_{{\mathbf{Q}_a^{-1}}}^2.
\label{eq:cwlsset}
\end{equation}
\begin{figure}[tbp]
\centering
\includegraphics[width=0.49\textwidth]{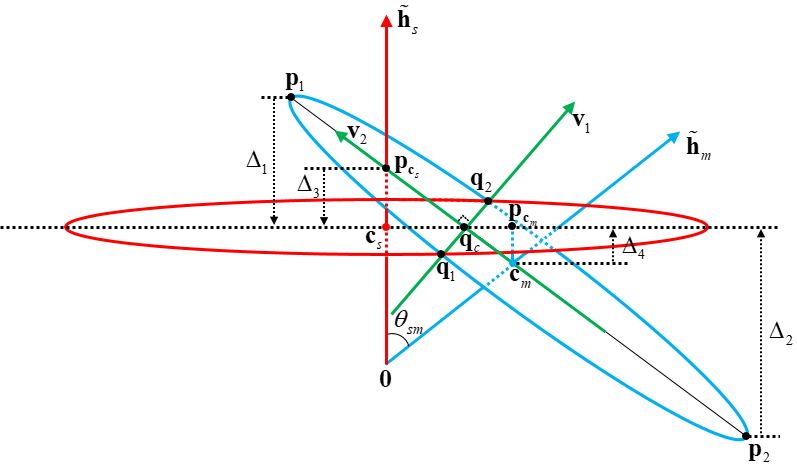}
\caption{Illustration of the intersection points.}
\label{fig:intersection3}
\end{figure}

To achieve a high-quality solution, a refinement of the coarse results is necessary. When the unambiguous phase $\bm{\phi}_a$ is available, we consider the following optimization
\begin{equation}
\mathop {\min }\limits_{{\mathbf{r}_a} \in {\mathbb{R}^{3}}, \left\| {\mathbf{r}_a} \right\|_2 = 1} \left\| {\left[ \!\!\! {\begin{array}{*{20}{c}}
  {\bm{\phi}_a  - d_a{\mathbf{H}}{\mathbf{r}_a}} \\ 
  {{\bm{\rho}}_a  - d_a{\mathbf{H}} {\mathbf{r}_a}} 
\end{array}} \!\!\! \right]} \right\|_{{\mathbf{Q}_a^{-1}}}^2,
\label{eq:srefin2}
\end{equation}
which leads to a general constraint least-squares solution.
Many strategies, generally iterative schemes, can be utilized to solve the nonlinear system in \eref{eq:srefin2}. 
We can directly utilize the ambiguous carrier phase as long as we substitute the following term in \eref{eq:srefin2}: 
\begin{equation}
\bm{\phi}_a   = \bm{\psi}_a + \text{round}( d_a{\mathbf{H}} {\mathbf{r}_a - \bm{\psi}_a}).
\label{eq:srefin3}
\end{equation}
We use the coarse results as the initial solutions and and achieve the final estimation based on \eref{eq:srefin2}--\eref{eq:srefin3}. The entire procedures are summarized in Algorithm~\ref{alg1}.

\begin{algorithm}[t]
\begin{algorithmic}[1]
\STATE ${\mathbb{S}}_{\mathbf{r}_a} = \varnothing$
\FORALL {Circle pairs}
\STATE Compute ${\bf{p}}_1$, ${\bf{p}}_2$, $\Delta_1$ and $\Delta_2$ using \eref{eq:int1}--\eref{eq:int4}.
\IF {$\Delta_1\Delta_2 < 0$}
\STATE Compute ${{\bf{q}}_1}$ and ${{\bf{q}}_2}$ using \eref{eq:int5}--\eref{eq:int10}.
\STATE ${\mathbb{S}}_{\mathbf{r}_a} = {\mathbb{S}}_{\mathbf{r}_a} \cup \{{{\bf{q}}_1},{{\bf{q}}_2}\}$.
\ELSE
\STATE When $|\Delta_1| < \delta_\Delta$, ${\mathbb{S}}_{\mathbf{r}_a} = {\mathbb{S}}_{\mathbf{r}_a} \cup \{{{\bf{p}}_1}\}$.
\STATE When $|\Delta_2| < \delta_\Delta$, ${\mathbb{S}}_{\mathbf{r}_a} = {\mathbb{S}}_{\mathbf{r}_a} \cup \{{{\bf{p}}_2}\}$.
\ENDIF
\ENDFOR
\end{algorithmic}
\caption{Procedure to search baseline vector candidates.}
\label{alg0}
\end{algorithm}

\begin{algorithm}[t]
\begin{algorithmic}[1]
\STATE  Determine ${\mathbb{S}}_{\mathbf{r}_a}$ based on Algorithm~\ref{alg0}.
\STATE Select the best $K$ estimations from ${\mathbb{S}}_{\mathbf{r}_a}$ using \eref{eq:cwlsset}.
\STATE Use the coarse results as the initial solutions.
\STATE Final results after refinement based on \eref{eq:srefin2}--\eref{eq:srefin3}.
\end{algorithmic}
\caption{C-WLS Single-baseline Attitude Determination.}
\label{alg1}
\end{algorithm}

\subsection{Multi-baseline Attitude Determination}
For multi-baseline configurations, we can estimate the rotation matrix by refining single-baseline solutions to avoid high complexity, rather than searching the solution directly. Single-baseline attitude determination applies only the baseline length constraint. To fully integrate the geometry constraints, we need to consider the relative direction of different baselines. This can be realized by checking the angle between estimations of baseline vectors. 


For each baseline, we keep the best $K$ estimations of the direction vector based on the intersection points and potential candidates. For the $a$-th and $k$-th baseline, we choose the estimations that satisfy
\begin{equation}
{\left| {\arccos \! \left( {{\mathbf{r}}_a^{\text{T}}{\mathbf{r}}_k} \right) - {\Theta _{ak}}} \right| < \delta_\Theta },
\label{eq:ang}
\end{equation}
where ${\Theta_{ak}}$ is the angle between two baselines, and $\delta_\Theta$ is the threshold of angle difference. The selected baseline vectors constitute the possible baseline matrix in the reference coordinate, that is
\begin{equation}{\mathbf{\overset{\lower0.5em\hbox{$\smash{\scriptscriptstyle\frown}$}} X}} = \left[ \!\! {\begin{array}{*{20}{c}}
  d_1{{\mathbf{r}}_1}& \cdots &d_{_\mathcal{A}}{{\mathbf{{r} }}_{_\mathcal{A}}} 
\end{array}} \!\! \right].
\label{eq:vtom}
\end{equation}
For $\mathcal{A}=2$, we can add a column to ${\mathbf{\tilde X}}$ and ${\mathbf{X}_b}$, that is
\begin{equation}
{\mathbf{\tilde X}} \!\!= \!\!\left[\! {{d_1{{\mathbf{r}}}_1}{\text{ }}{d_2{{\mathbf{r}}}_2}{\text{ }}} \frac{{{\mathbf{r}}_1 \!\times\! {\mathbf{r}}_2}}{{{{\left\| {{\mathbf{r}}_1 \!\times \!{\mathbf{ r}}_2} \right\|}_2}}} \!\right],
{{\mathbf{X}}_b} \!\!= \!\!\left[\! {{\mathbf{x}}_b^1{\text{ }}{\mathbf{x}}_b^2{\text{ }}\frac{{{\mathbf{x}}_b^1 \!\times\! {\mathbf{x}}_b^2}}{{{{\left\|{{\mathbf{x}}_b^1 \!\times \!{\mathbf{x}}_b^2} \right\|}_2}}}} \!\right].
\end{equation}
According to the relationship between the baseline matrix in the body frame and reference frame, the coarse solution of the rotation matrix can be calculated using
\begin{equation}
{\mathbf{\tilde R}} = {\mathbf{\tilde X}}{\mathbf{X}}_b^{\text{T}}{\left( {{{\mathbf{X}}_b}{\mathbf{X}}_b^{\text{T}}} \right)^{ - 1}}.
\label{eq:conv}
\end{equation}

\begin{algorithm}[t]
\begin{algorithmic}[1]
\STATE Find the best $K$ estimations for each baseline using step 1--2 of Algorithm~\ref{alg1}.
\STATE Keep the estimations that satisfy the angle criteria \eref{eq:ang} for every baseline pair.
\STATE Combine pointing vectors to create the baseline matrix using \eref{eq:vtom}.
\STATE Convert the baseline matrix to a rotation matrix using \eref{eq:conv}.
\STATE Initialize the orthogonal rotation matrix using Wahba’s problem as in \eref{eq:wahhaq}.
\STATE Final results after refinement via \eref{eq:refin1}--\eref{eq:refin3}.
\end{algorithmic}
\caption{C-WLS Multi-baseline Attitude Determination.}
\label{alg2}
\end{algorithm}
To satisfy the antenna array constraints, we have to transform the rotation matrix to an orthogonal matrix, that is, to solve the following optimization problem:
\begin{equation}
{\mathbf{\overset{\lower0.5em\hbox{$\smash{\scriptscriptstyle\frown}$}} R}} = \arg\mathop {\min }\limits_{{\mathbf{R}} \in {\mathbb{O}^{3 \times 3}}} {\left\| {{\mathbf{R}} - {\mathbf{\tilde R}}} \right\|_2}.
\label{eq:wahhaq}
\end{equation}
This is the Wahba’s problem for which an analytical solution exists \cite{wahba1965problem, markley1988attitude}. The solution goes as follows. First, apply singular-value decomposition to the coarse rotation matrix
\begin{equation}
{\mathbf{\tilde R}} = {\bf{\tilde U}}{\bf{\tilde \Sigma}} {{\bf{\tilde V}}^{\text{T}}},
\label{eq:wahhaq1}
\end{equation}
where ${\bf{\tilde \Sigma}}$ is a diagonal matrix in which the diagonal entries are the singular
values of ${\mathbf{\tilde R}}$, and the columns of ${\bf{\tilde U}}$ and ${\bf{\tilde V}}$ are the left and right singular
vectors of ${\mathbf{\tilde R}}$, respectively.
Then a solution to (\ref{eq:wahhaq}) can be obtained as 
\begin{equation}
{\mathbf{\overset{\lower0.5em\hbox{$\smash{\scriptscriptstyle\frown}$}} R}} = {\mathbf{\tilde U}}\left[\!\!\! {\begin{array}{*{20}{c}}
  1&0&0 \\ 
  0&1&0 \\ 
  0&0&{\det ({\mathbf{\tilde U}})\det ({\mathbf{\tilde V}})} 
\end{array}} \!\!\! \right]{{\mathbf{\tilde V}}^{\text{T}}}.
\label{eq:wahhaq2}
\end{equation}
It can be readily seen that ${\mathbf{\overset{\lower0.5em\hbox{$\smash{\scriptscriptstyle\frown}$}} R}}$ is an orthogonal matrix, i.e, it satisfies all the required constraints. However, this is only a sub-optimal estimate of the rotation matrix; hence, a refinement is required. 

Similar to the single-baseline set-up, if the unambiguous phase $\bm{\Phi}$ is known, we can consider the following optimization
\begin{equation}
\mathop {\min }\limits_{{\mathbf{R}} \in {\mathbb{O}^{3 \times q}}} \left\| {\left[ \!\!{\begin{array}{*{20}{c}}
  {\operatorname{vec} \! \left(  {\bm{\Phi}  - {\mathbf{H}}{\mathbf{R}{\mathbf{X}}_b}} \right)} \\ 
  \operatorname{vec} \! \left({{\bf{P}}  - {\mathbf{H}} {\mathbf{R}} {\mathbf{X}_b}} \right)
\end{array}}\!\! \right]} \right\|_{{\mathbf{Q}^{-1}}}^2.
\label{eq:refin1}
\end{equation}
Again, we directly take advantage of the ambiguous phase by using ${\mathbf{\overset{\lower0.5em\hbox{$\smash{\scriptscriptstyle\frown}$}} R}}$ as the initial solution and applying the substitution
\begin{equation}
\bm{\Phi}  = \bm{\Psi} + \text{round}\left({\mathbf{H}} {\mathbf{R}{\mathbf{X}}_b} -\bm{\Psi} \right).
\label{eq:refin3}
\end{equation}
The entire process to estimate the rotation matrix and resolve the carrier-phase ambiguity is summarized in Algorithm~\ref{alg2}.

\section{Performance Evaluation}
\label{sec:result}
In this section, we first present simulation results. Next, we evaluate the performance of the proposed method experimentally. Our tests focus on the most challenging single-frequency and single-epoch scenarios, with only GPS constellation utilized. To demonstrate the performance of the proposed approach, we benchmark against the AFM-based \cite{LiJul2004} and LAMBDA-based search algorithms (the C-LAMBDA method \cite{teunissen2010integer} for single-baseline set-ups and the MC-LAMBDA method \cite{giorgi2010testing} in multi-baseline scenarios). A 10-degree cut-off elevation mask is used in all cases, as is usually recommended to protect against severe multi-path effects.


\subsection{Simulation Results}
Simulations are implemented based on the Visual software \cite{verhagen2006visual} using the assumed antenna position and real GPS constellation information on November 7, 2021. We evaluate the performance of the proposed approach under different noise levels by adding Gaussian noise with different standard deviations to both the pseudo-range and carrier-phase measurements. We also examine the feasibility of the proposed method under adverse satellite geometry conditions, namely scenarios with a poor position dilution of precision (PDOP) due to limited satellite visibility. A pseudo-range to carrier-phase variance ratio ${\sigma_{\rho}^2}/{\sigma_{\psi}^2} = 10^4$ is adopted \cite{giorgi2013low}. 
We consider setups involving one, two, and three baselines with a 1-meter baseline length. In the multi-baseline cases, the baselines are perpendicular to each other, so the baseline matrices in the body frame are 
\begin{equation}
{{\bf{X}}_b^1} = \left[ {\begin{array}{*{20}{c}}
{1}&{0}\\
0&{1}
\end{array}} \right],
\quad {{\bf{X}}_b^2} = \left[ {\begin{array}{*{20}{c}}
1&0&0\\
0&1&0\\
0&0&1
\end{array}} \right].
\label{xb}
\end{equation}
For each simulated scenario, the tests are repeated $10^4$ times. Each time the attitude angles are generated randomly.

\begin{table*}[htbp]
\centering
  \caption{Success rate (\%) versus phase noise $\sigma_{\psi} (\text{mm})$, number of satellites (\#Sat), and number of baselines.\\C-LAMBDA/MC-LAMBDA \\AFM  \\Proposed}
  \setlength{\tabcolsep}{1.2mm}{
    \begin{tabular}{|c|ccccc|ccccc|ccccc|}
    \hline
     & \multicolumn{5}{c|}{Single Baseline} & \multicolumn{5}{c|}{Two Baselines} & \multicolumn{5}{c|}{Three Baselines}\\
    \hline
    \diagbox{\#Sat}{$\sigma_{\psi}$}  &9 &7 &5 &3 &1       &9 &7 &5 &3 &1   &9 &7 &5 &3 &1\\
    \hline
    \multirow{3}[2]{*}{4}

           &6.43  &\textbf{9.11}  &\textbf{14.54} &\textbf{29.23}  &\textbf{82.74} &0 &0.11 &1.72 &22.43 &\textbf{99.02} &0 &0 &0 &0.25 &88.52\\
           &3.33  &4.01 &6.23  &9.37 &23.78 &2.08 &3.12 &6.55 &14.56 &37.75 &15.67 &33.22 &62.61 &94.09 &99.01\\
           &\textbf{6.48} &\textbf{9.11} &\textbf{14.54} &\textbf{29.23} &\textbf{82.74} &\textbf{3.43} &\textbf{8.35} &\textbf{19.45} &\textbf{54.87} &98.81 &\textbf{22.02} &\textbf{43.18} &\textbf{72.37} &\textbf{96.08} &\textbf{99.76}\\
    \hline
    \multirow{3}[2]{*}{5}

           &\textbf{13.47}  &21.46  &\textbf{37.51}  &\textbf{70.11}  &\textbf{99.11} &1.22 &6.95 &40.27 &96.09 &\textbf{99.97} &0 &0.15 &1.13 &43.54 &\textbf{100}\\
           &7.24 &11.21 &19.45 &39.72 &83.94 &7.23 &21.09 &45.76 &78.22 &95.21 &52.88 &81.09 &95.67 &99.31 &99.99\\
           &\textbf{13.47}  &\textbf{21.47} &\textbf{37.51} &\textbf{70.11} &\textbf{99.11} &\textbf{16.54} &\textbf{36.82} &\textbf{69.44} &\textbf{96.27} &\textbf{99.97} &\textbf{66.47} &\textbf{89.45} &\textbf{99.1} &\textbf{100} &\textbf{100}\\          
    \hline
    \multirow{3}[2]{*}{6} 

           &\textbf{26.97} &\textbf{43.88} &\textbf{70.40} &\textbf{94.48} &\textbf{100} &16.35 &50.46 &93.41 &\textbf{99.85} &\textbf{100} &0.45 &4.17 &43.76 &\textbf{100} &\textbf{100}\\
           &17.97 &31.71 &51.74 &80.59 &98.83 &26.55 &58.77 &76.33 &98.31 &99.23 &84.12 &94.36 &98.59 &99.92 &\textbf{100}\\
           &\textbf{26.97} &\textbf{43.88} &\textbf{70.40}  &\textbf{94.48} &\textbf{100} &\textbf{47.95} &\textbf{76.96} &\textbf{96.16} &\textbf{99.85} &\textbf{100} &\textbf{93.19} &\textbf{99.36} &\textbf{99.99} &\textbf{100} &\textbf{100}\\      
    \hline
    \multirow{3}[2]{*}{7}  

           &\textbf{46.42} &69.04 &\textbf{91.22} &\textbf{99.38} &\textbf{100} &50.59 &90.13 &\textbf{99.59} &\textbf{99.99} &\textbf{100} &6.18 &42.16 &97.87 &\textbf{100} &\textbf{100}\\
           &37.53 &57.11 &79.58 &95.73 &99.82 &48.87 &78.59 &92.67 &99.89 &\textbf{100} &88.31 &97.31 &99.31 &99.99 &\textbf{100}\\
           &\textbf{46.42} &\textbf{69.05} &\textbf{91.22} &\textbf{99.38} &\textbf{100} &\textbf{78.16} &\textbf{95.12} &\textbf{99.59} &\textbf{99.99} &\textbf{100} &\textbf{99.24} &\textbf{99.97} &\textbf{100} &\textbf{100} &\textbf{100}\\
    \hline
    \multirow{3}[2]{*}{8}  

           &67.63 &87.26 &\textbf{98.01} &\textbf{99.96} &\textbf{100} &85.51 &\textbf{99.16} &\textbf{99.96} &\textbf{100} &\textbf{100} &37.39 &89.89 &\textbf{100} &\textbf{100} &\textbf{100}\\
           &58.02 &78.06 &92.68 &98.69 &99.94 &85.01 &91.13 &98.11 &99.89 &\textbf{100} &95.13 &98.78 &99.95 &\textbf{100} &\textbf{100}\\
           &\textbf{67.64} &\textbf{87.27} &\textbf{98.01} &\textbf{99.96} &\textbf{100} &\textbf{93.29} &\textbf{99.16} &\textbf{99.96} &\textbf{100} &\textbf{100} &\textbf{99.84} &\textbf{99.98} &\textbf{100} &\textbf{100} &\textbf{100}\\
    \hline
    \end{tabular}%
    }
  \label{tab:suc2}%
\end{table*}%
\begin{table*}[htbp]
\centering
  \caption{Average number of integers in the search space of the C-LAMBDA or MC-LAMBDA method versus phase noise $\sigma_{\psi} (\text{mm})$, number of satellites (\#Sat), and number of baselines.}
\setlength{\tabcolsep}{1.0mm}{
\resizebox{\textwidth}{!}{
  \scalebox{1.0}{
    \begin{tabular}{|c|ccccc|ccccc|ccccc|}
    \hline
     & \multicolumn{5}{c|}{Single Baseline} & \multicolumn{5}{c|}{Two Baselines} & \multicolumn{5}{c|}{Three Baselines}\\
    \hline
    \diagbox{\#Sat}{$\sigma_{\psi}$}  &9 &7 &5 &3 &1       &9 &7 &5 &3 &1  &9 &7 &5 &3 &1\\
    \hline
    4
    &14203.3 &6619.9 &2536.8 &595.7 &40.5  &99916.2 &99965.5 &99675.1 &95971.5 &4175.6 &100000 &100000 &100000 &100000 &73607.3\\
     
    \hline
    5
    &3656.1 &1483.7 &492.3 &94.6 &3.8 &99626.3 &99255.3 &89286.8 &25015.6   &12.1 &100000 &100000 &100000 &95734.4 &49.9\\ 
          
    \hline
    6
    &1239.4 &463.4 &127.1 &16.1 &2.1  &96414.6 &81804.9 &33112.7 &599.7 &2.0 &100000 &100000 &93842.6 &9697.9 &2.1\\
          
    \hline
    7  
    &533.4 &166.8 &34.5 &3.8 &2.0 &80451.1 &36197.0 &2107.1 &10.6 &2.0 &99839.2 &93376.0 &32230.0 &40.9 &2.0\\
     
    \hline
    8  
    &216.4 &54.8 &8.6 &2.3  &2.0 &40763.2 &6298.8 &76.3 &2.3 &2.0 &93682.6 &49282.2 &1020.4 &2.7 &2.0\\
      
    \hline
    \end{tabular}%
    }
    }
    }
  \label{tab:num1}%
\end{table*}%
\begin{table}[htbp]
\caption{Computational Complexity of Different Algorithms.}
\label{tab:complex}
\centering
    \begin{tabular}{|c |c|}
    \hline
    \textbf{Algorithms}  &Complexity \\
    \hline
        Proposed &$O \!\! \left( {\mathcal{A}^3} \frac{\mathcal{S}^3\left(\mathcal{S}-1\right)}{2} K_w^2 \right)$ \\
    \hline
    AFM &$O \! \left({\mathcal{A}} {\mathcal{S}} K_a^n \right)$\\
    \hline
        C-LAMBDA&
    \multirow{2}{*}{$O \! \left({\mathcal{A}^2} {\mathcal{S}^2} K_sK_l \right)$} \\
      MC-LAMBDA &\\
       \hline
    \end{tabular}
\end{table}

Table~\ref{tab:suc2} reports the success rates of ambiguity resolution for the proposed method and the benchmark methods. These results are obtained for a number of different scenarios involving different numbers of baselines, different numbers of satellites, and various noise levels. The success rate is calculated as the percentage of trials in which all the integer ambiguities are retrieved correctly. Note that success cases can provide accurate attitude estimation, whereas failure cases tend to produce outliers or estimates that are meaningless. This underscores the pertinence of success rate as a performance indicator.

\begin{figure}[htbp]
\centering
\subfigure[$\sigma_{\psi} = 1\text{ mm}$] { \label{fig:1a} 
\includegraphics[width=0.46\columnwidth, height=0.33\columnwidth]{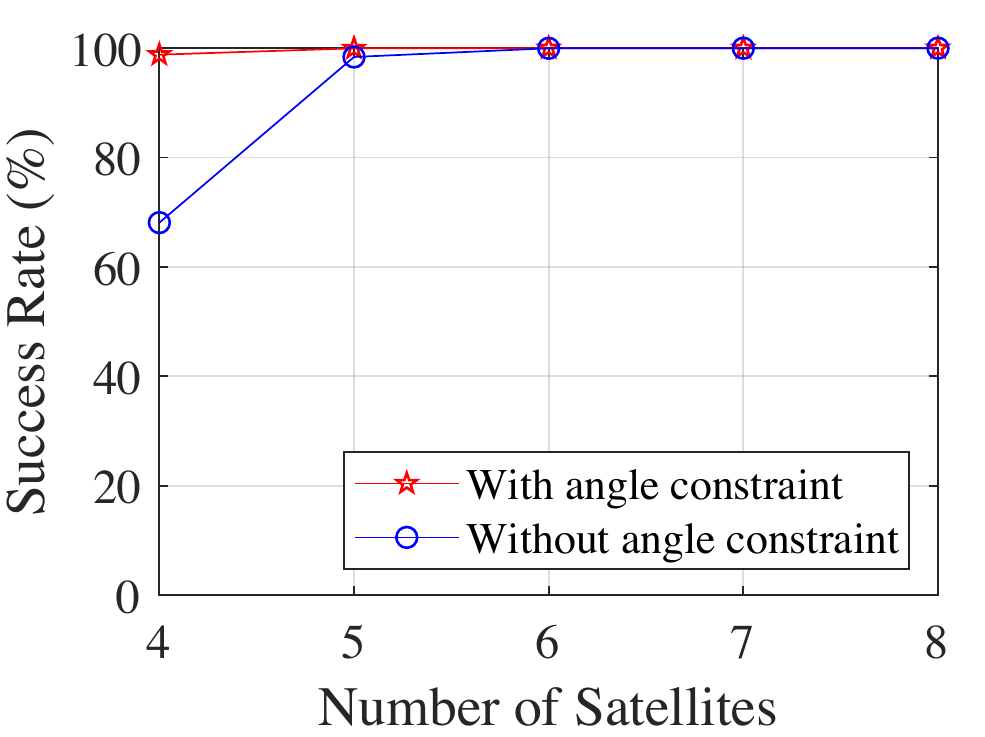}
} 
\subfigure[$\sigma_{\psi} = 3\text{ mm}$] { \label{fig:1b} 
\includegraphics[width=0.46\columnwidth, height=0.33\columnwidth]{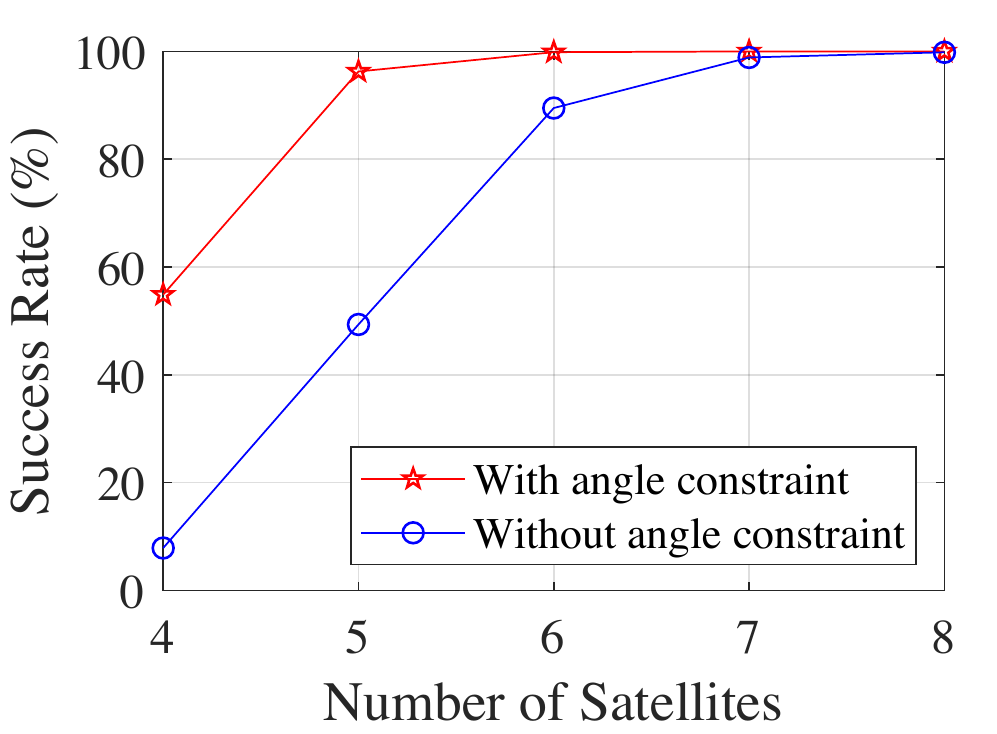} 
} 

\subfigure[$\sigma_{\psi} = 5\text{ mm}$] { \label{fig:1c} 
\includegraphics[width=0.46\columnwidth, height=0.33\columnwidth]{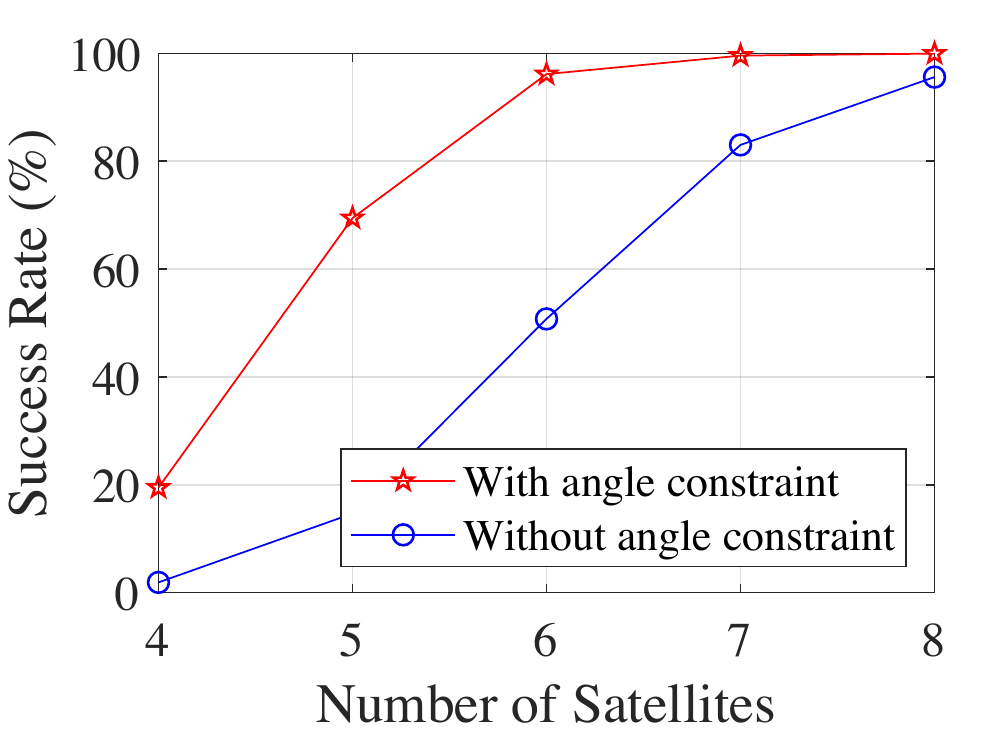} 
} 
\subfigure[$\sigma_{\psi} = 7\text{ mm}$] { \label{fig:1d} 
\includegraphics[width=0.46\columnwidth, height=0.33\columnwidth]{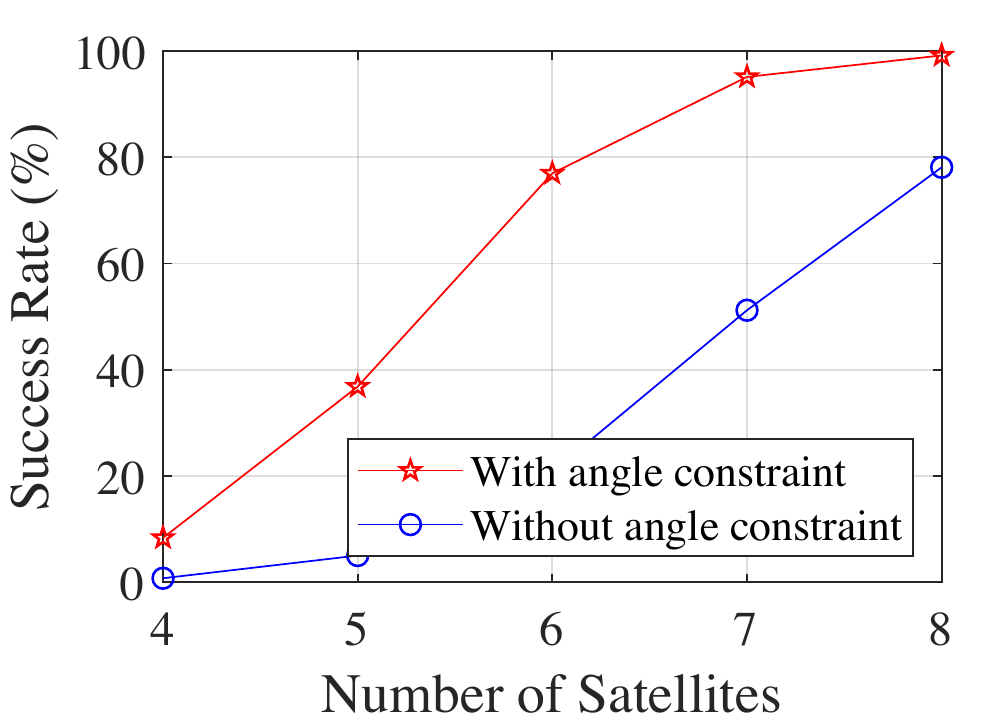} 
} 
\caption{Success rate with and without the angle constraint for dual-baseline set-up.}
\label{fig:angle}
\end{figure}

Based on the search-and-shrink or search-and-expand algorithm, the C-LAMBDA and MC-LAMBDA methods can adaptively adjust the search space in the integer domain. Theoretically, they could find the global optimum of the optimization \eref{eq:cls}. However, the performance of the C-LAMBDA and MC-LAMBDA methods dramatically depends on the accuracy of the float solutions. In challenging scenarios, the inadequacy of the float solutions results in an extensive set of potential integer vectors to be evaluated by these two approaches. The set can grow even larger as the number of baselines increases and the dimensionality of the ambiguity vectors/matrices grows. To avoid computational issues due to prohibitively large search spaces, we force the adaptive adjustment strategies of the C-LAMBDA and MC-LAMBDA methods to stop after a search limit of $10^5$ integer searches.

As shown in Table~\ref{tab:suc2}, the proposed approach and the C-LAMBDA method have almost the same success rates in single-baseline configurations, both outperforming the AFM method. In the multiple baseline case, the proposed method utilizes the angle constraint to evaluate and improve the single-baseline estimates; hence, its performance improves as more baselines are used. This should be the case, more or less, with the MC-LAMBDA method if a search limit is not imposed. 
However, without the search limit, an unacceptable computational complexity results in the case of the MC-LAMBDA algorithm. As shown in \tref{tab:num1}, the C-LAMBDA and MC-LAMBDA methods need to evaluate just a few integer candidates in ideal scenarios. However, an enormous number of integers need to be evaluated in many challenging setups. 
The application of a search limit to the C-LAMBDA and MC-LAMBDA methods explains their low success rate in Table~\ref{tab:suc2}, but it makes the two algorithms computationally feasible.



Fig.~\ref{fig:angle} demonstrates how the angle constraint affects the success rate for a dual-baseline set-up. Table~\ref{tab:suc2} and Fig.~\ref{fig:angle} demonstrate that the proposed approach provides the highest success rate in almost all scenarios, which illustrates that this method is remarkably effective. In contrast with the AFM-based methods, the C-WLS approach jointly utilizes pseudo-range and carrier-phase measurements instead of only carrier-phase observation. Besides, we strengthen the proposed C-WLS model by incorporating the high correlations of double-difference observations. On the other hand, the proposed approach outperforms the C-LAMBDA and MC-LAMBDA methods, which is attributed to the additional constraint on the carrier-phase residual errors and the proposed search strategy. The superiority of the proposed method is more pronounced in scenarios with fewer satellites and higher noise levels.

\tref{tab:complex} summarizes the computational complexity of different algorithms. We characterize the computational complexity based on the required number of operations. For simplicity, we count each addition, subtraction, multiplication, division, and square root calculation as one arithmetic operation. The complexity of each method depends on the two parameters ${\mathcal{A}}$ and ${\mathcal{S}}$ in addition to a third algorithm-specific parameter $K$. The computational complexity of the proposed method depends on the parameter $K_w$, which represents the average number of possible integer ambiguities for each element of ${\bm{\Psi}}$. For the AFM method, we assume that an identical search step is applied for each Euler angle. The parameter $K_{a}$ represents the number of potential attitude angles, which is practically much greater than $K_w$. The symbol $n$ denotes the number of Euler angles, i.e., $n = 2$ for a single baseline or $n = 3$ for multiple baselines. Since ${\mathcal{A}}, {\mathcal{S}}, K_w \ll K_a$, the complexity of the proposed method is much lower than that of the AFM method, especially in multi-baseline scenarios. Note that $K_w$ and $K_{a}$ are proportionate to the baseline length. Hence, the proposed approach and the AFM method are more adequate for short-baseline set-ups (meter level). The C-LAMBDA and MC-LAMBDA methods include three steps: shrink the search space, enumerate the integer vectors/matrices, and minimize the objective function \cite{giorgi2008search}. The first step dominates the complexity of these methods as it accounts for at least 60\% of the total arithmetic operations  \cite{giorgi2008search}. The parameters $K_s$ and ${\mathcal{A}^2} {\mathcal{S}^2}K_l$ denote the number of iterations in the shrinking process and the required operations in each iteration, respectively. $K_s$, as shown in \tref{tab:num1}, is extremely sensitive to various parameters, such as ${\mathcal{A}}$, ${\mathcal{S}}$, $\mathbf{H}$, $\bm{\Xi}$, and $\bm{\Pi}$; nevertheless, $K_l$ depends on the value of $K_s$. In ideal environments, the product $K_sK_l$ can be smaller than ${\mathcal{A}} \frac{\mathcal{S}\left(\mathcal{S}-1\right)}{2} K_w^2$ such that the C-LAMBDA and MC-LAMBDA algorithms are more computationally efficient than the proposed approach. However, in challenging scenarios such as those with poor PDOP, high noise levels, and severe multi-path, $K_sK_l$ can be much greater than ${\mathcal{A}} \frac{\mathcal{S}\left(\mathcal{S}-1\right)}{2} K_w^2$ ($K_sK_l$ could reach ${10\mathcal{A}} \frac{\mathcal{S}\left(\mathcal{S}-1\right)}{2} K_w^2$, $100{\mathcal{A}} \frac{\mathcal{S}\left(\mathcal{S}-1\right)}{2} K_w^2$, or even more), making the computational complexity of the C-LAMBDA and MC-LAMBDA methods being far worse than that of the proposed method. For instance, consider the antenna geometry described by ${{\mathbf{X}}_b^1}$ in \eref{xb}, i.e., ${\mathcal{A}} =2$, with $\sigma_{\psi} = 3 \text{ mm}$ and $\sigma_{\rho} = 30 \text{ cm}$. For ${\mathcal{S}}= 6$, we have $K_w = 10.19$, $K_a = 180$, $K_s = 10.60$, and $K_l = 35.03$, then ${\mathcal{A}^2} {\mathcal{S}^2}K_sK_l < {\mathcal{A}^3} \frac{\mathcal{S}^3\left(\mathcal{S}-1\right)}{2} K_w^2 < {\mathcal{A}} {\mathcal{S}}K_a^3$. For ${\mathcal{S}}= 4$, we have $K_w = 8.33$, $K_a = 180$, $K_s = 25015.6$, and $K_l = 14.01$, then we obtain ${\mathcal{A}^3} \frac{\mathcal{S}^3\left(\mathcal{S}-1\right)}{2} K_w^2 < {\mathcal{A}^2} {\mathcal{S}^2}K_sK_l < {\mathcal{A}} {\mathcal{S}}K_a^3$.

\begin{figure*}[htbp]
\centering
\subfigure[Yaw] { \label{fig:2a} 
\includegraphics[width=0.62\columnwidth, height=0.46\columnwidth]{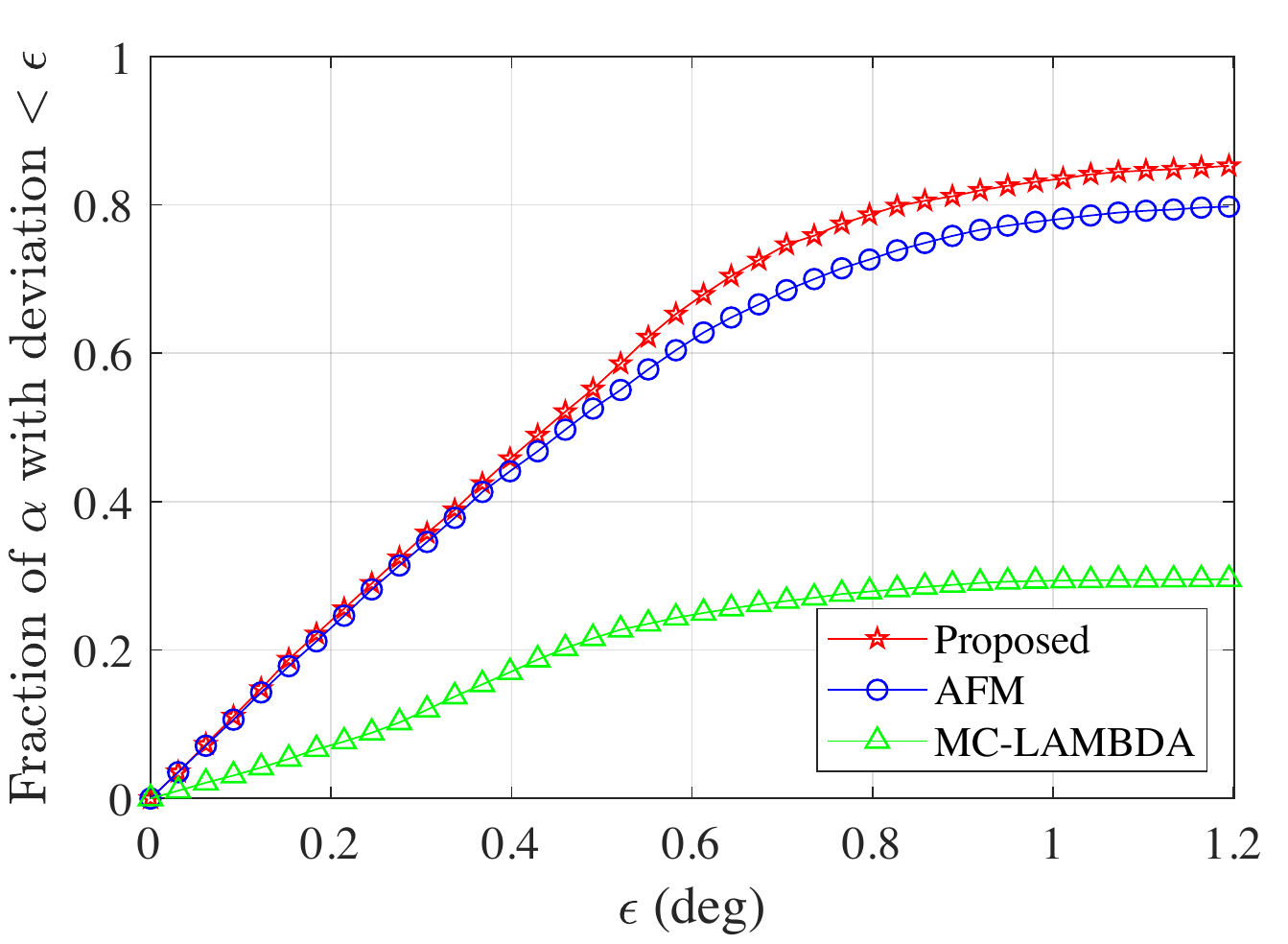}
} 
\subfigure[Pitch] { \label{fig:2b} 
\includegraphics[width=0.62\columnwidth, height=0.46\columnwidth]{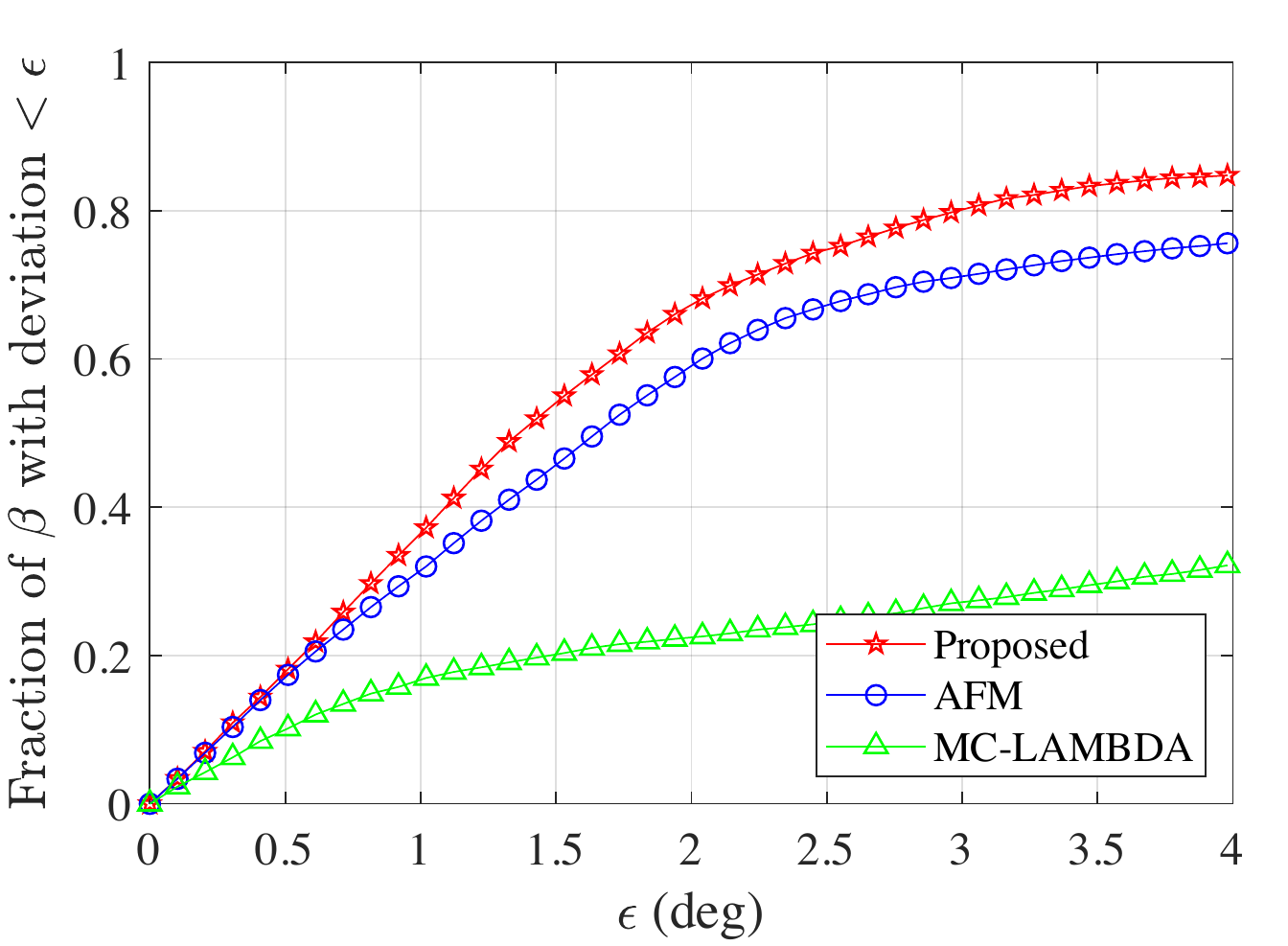} 
} 
\subfigure[Roll] { \label{fig:2c} 
\includegraphics[width=0.62\columnwidth, height=0.46\columnwidth]{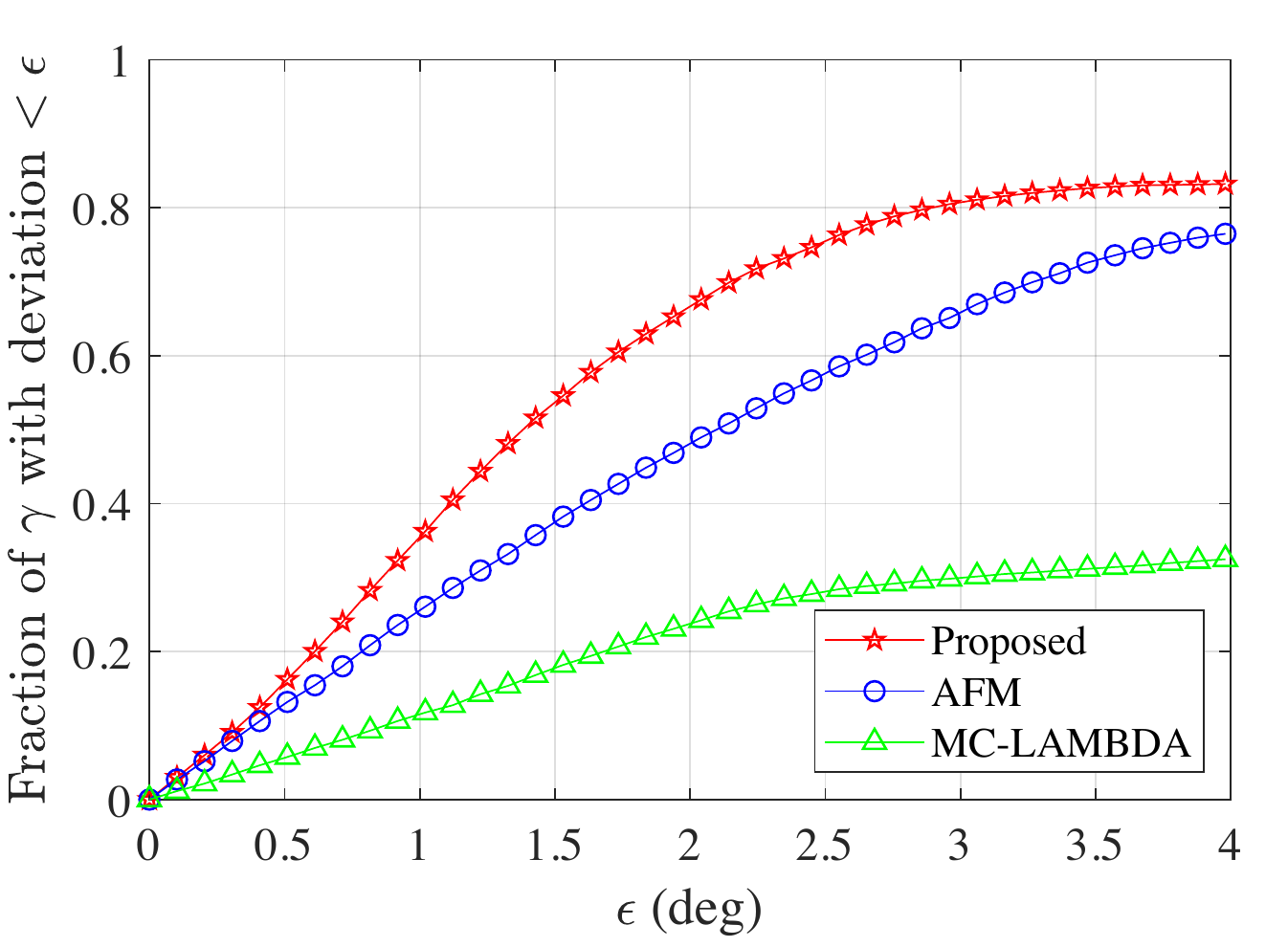} 
} 
\caption{Attitude angle error distribution for data 1.} 
\label{fig:result1} 
\end{figure*}
\begin{figure*}[htbp]
\centering
\subfigure[Yaw] { \label{fig:3a} 
\includegraphics[width=0.62\columnwidth, height=0.46\columnwidth]{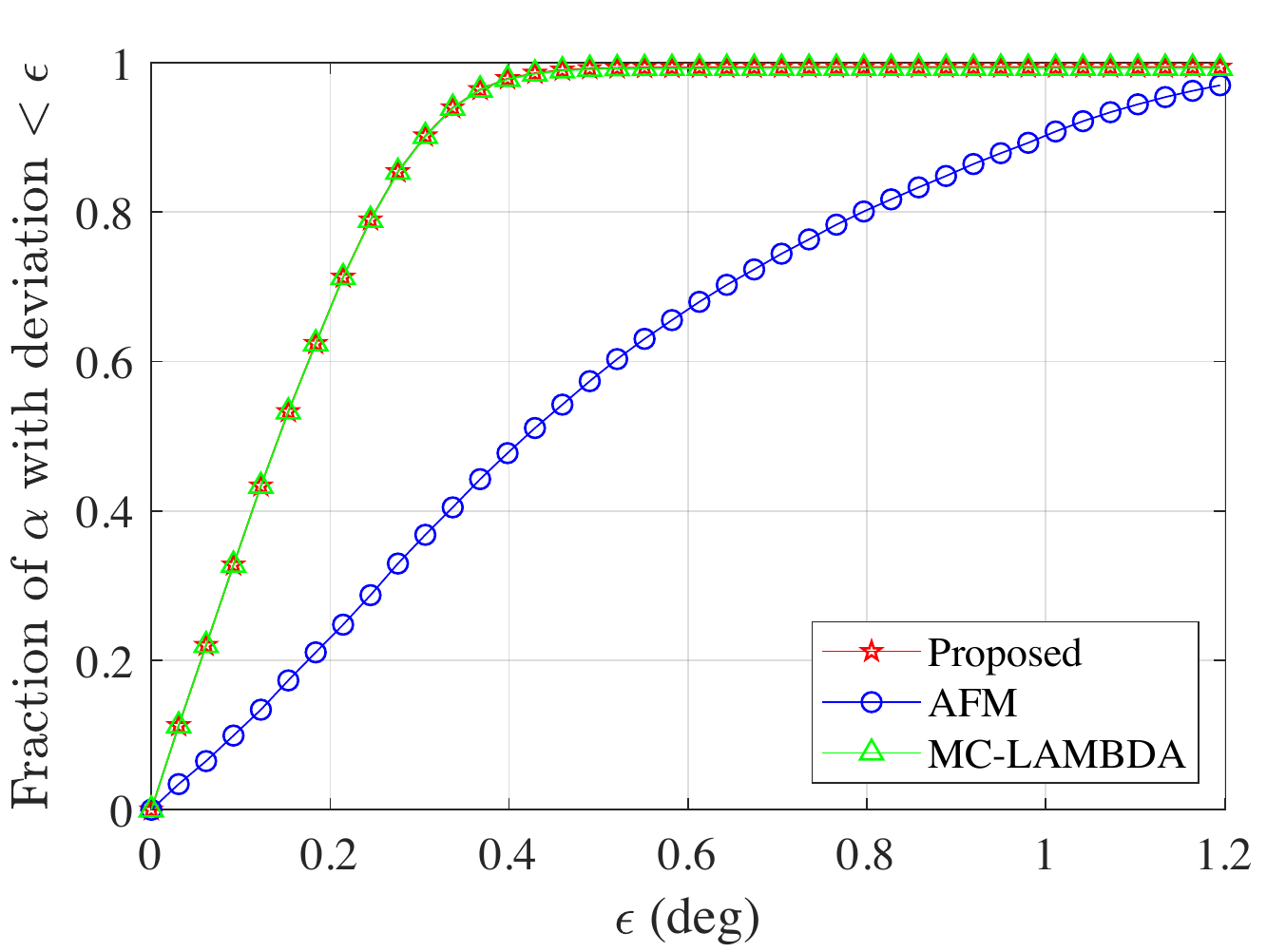}
} 
\subfigure[Pitch] { \label{fig:3b} 
\includegraphics[width=0.62\columnwidth, height=0.46\columnwidth]{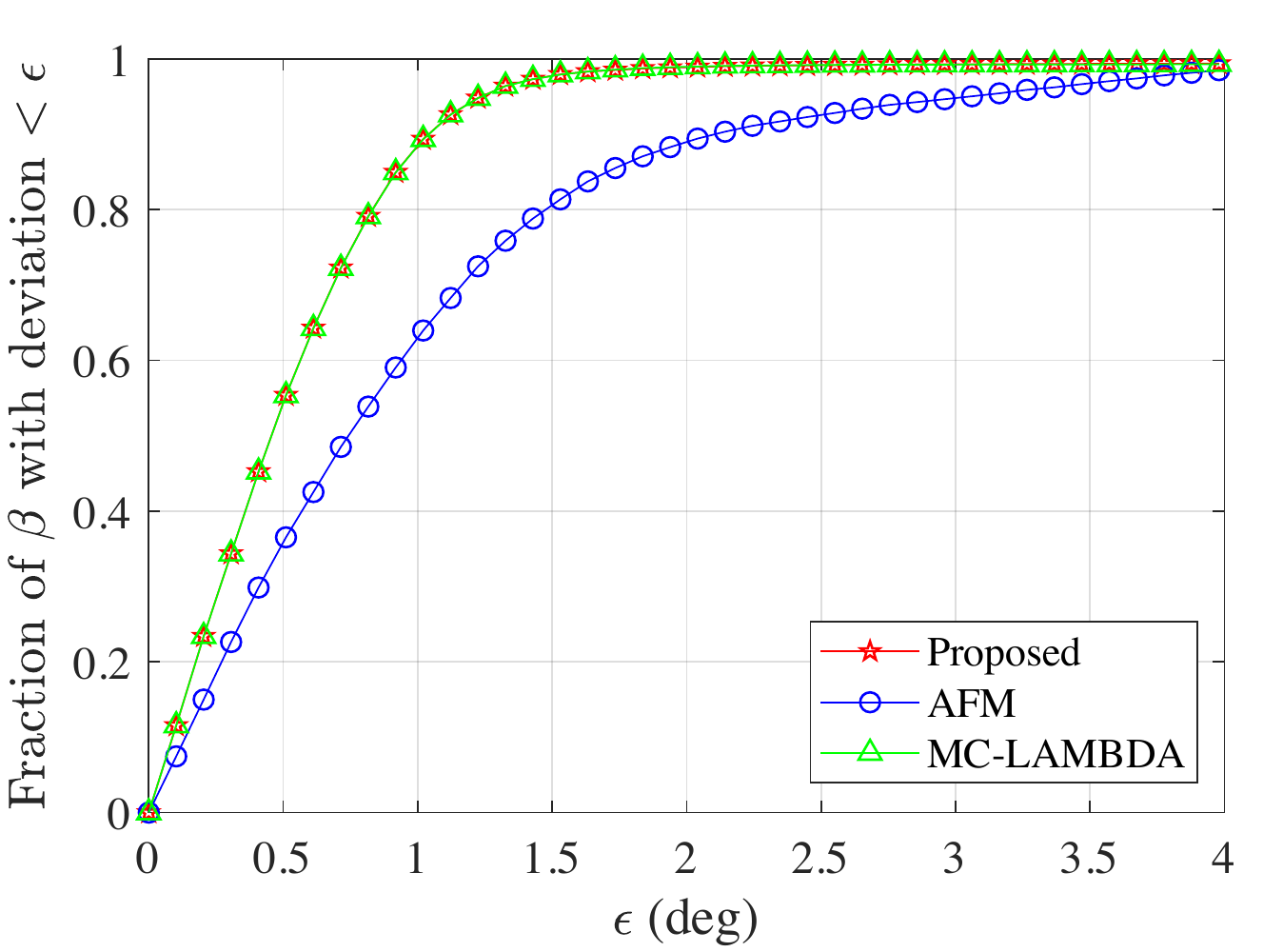} 
} 
\subfigure[Roll] { \label{fig:3c} 
\includegraphics[width=0.62\columnwidth, height=0.46\columnwidth]{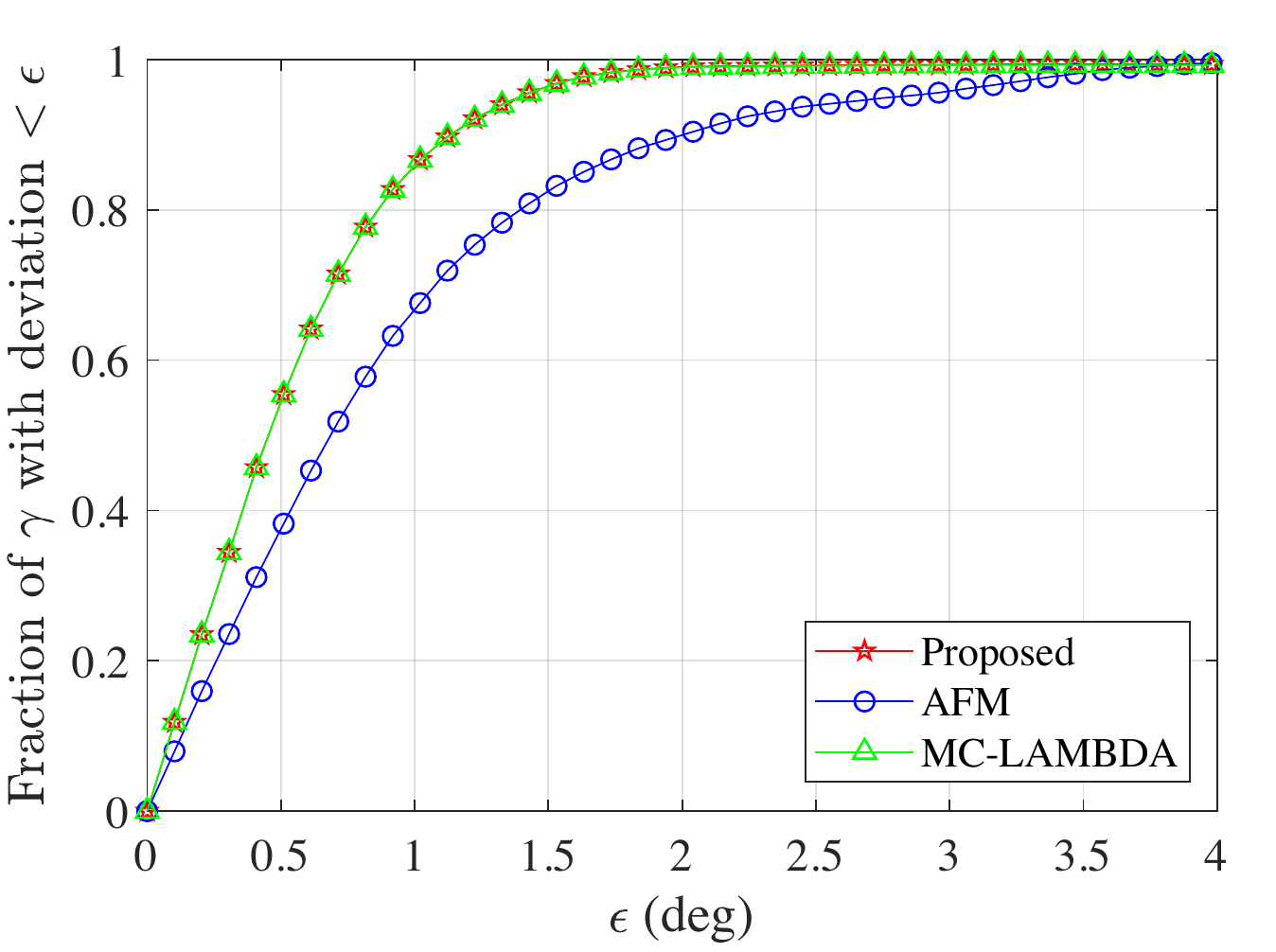} 
} 
\caption{Attitude angle error distribution for data 2.} 
\label{fig:result2} 
\end{figure*}

\subsection{Experimental Results}
In this section, we provide performance evaluation based on experimental tests. Three ANAVS multi-sensor modules were used \cite{anavs}, which include a multi-GNSS receiver, an inertial measurement unit (IMU), and a barometer. To be able to measure the ground truth with reasonable accuracy, we consider two static experiments performed on the campus of King Abdullah University of Science and Technology, Thuwal, Saudi Arabia. The receivers were firmly attached to a fixed platform and remained stationary throughout the experiments. The baseline matrix in the body frame used in both tests is described by
\[{{\mathbf{X}}_b^3} = \left[ {\begin{array}{*{20}{c}}
{0.63}&{0.315}\\
0&{0.545}
\end{array}} \right].\]

Two datasets were collected with 5 Hz sampling on 6 April 2021 (between 18:05 and 18:35, UTC time) and 25 May 2021 (from 20:34 to 21:50, UTC time), respectively. The first experiment collected around 9000 epochs of data within 30 minutes, and the other test collected 25500 epochs of observations. We apply the elevation-dependent model to characterize the GNSS observation variance components \cite{jin1996relationship} and consider the correlation introduced by differencing operations. Note that in our comparison of various methods, we utilize only GPS observations, with the number of tracked satellites and the corresponding PDOP values shown in Fig.~\ref{fig:PDOP}. In contrast, the ground truth is estimated by leveraging observations from two GNSS constellations (GPS and GLONASS), IMU measurements, and barometer data integrated over the whole observation time. 




\begin{figure}[tbp]
\centering
\subfigure[First experiment] { \label{fig:pdopa} 
\includegraphics[width=0.465\columnwidth]{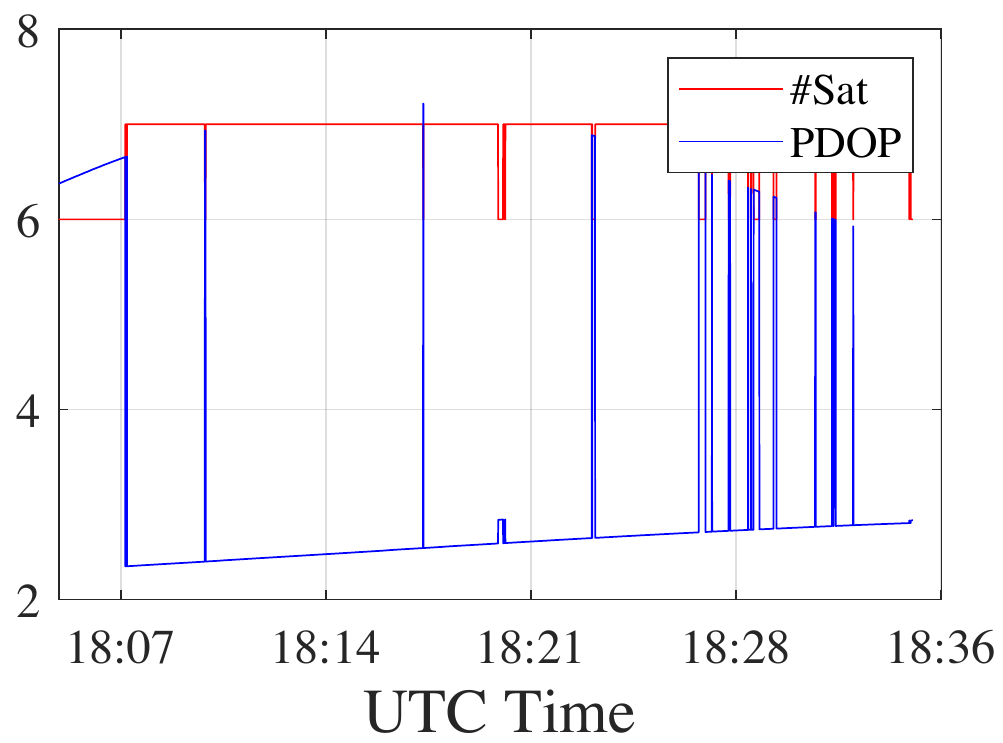}
} 
\subfigure[Second experiment] { \label{fig:pdopb} 
\includegraphics[width=0.465\columnwidth]{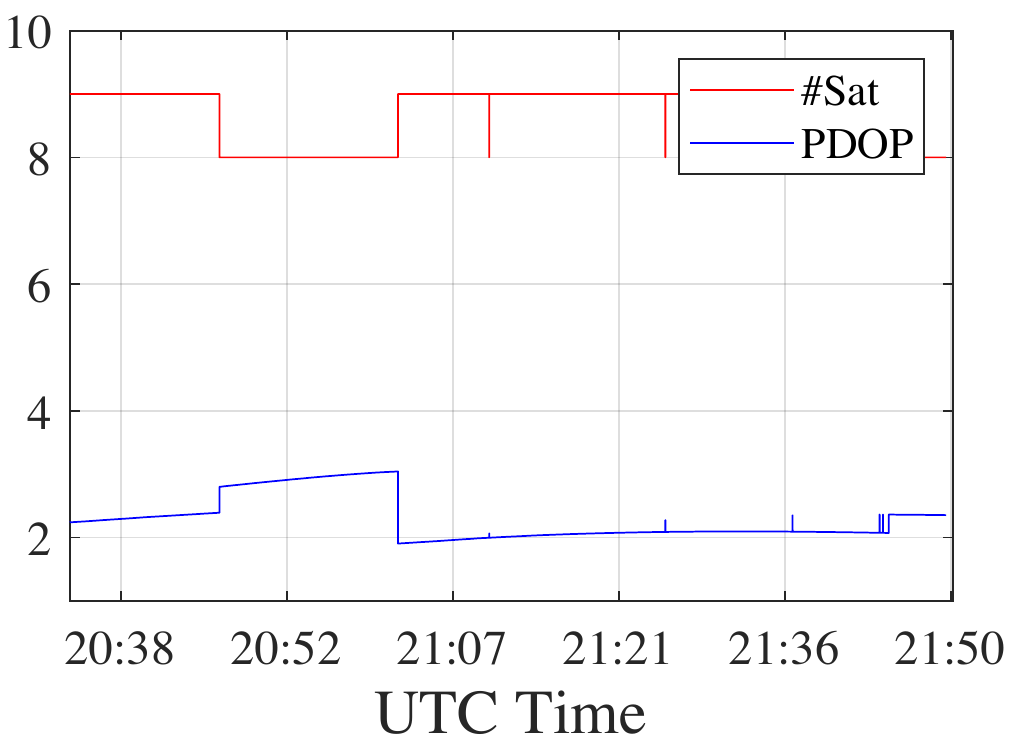} 
} 
\caption{The number of tracked satellites and PDOP values.}
\label{fig:PDOP}
\end{figure}

Table~\ref{tab:suc3} summarizes the experimental results by listing the success rate and the root-mean-square error (RMSE) of attitude angles for the two experiments. Only success cases are considered in calculating the RMSE to eliminate the contribution of outliers. As stated earlier, the constrained search space (at most $10^5$ integers) in the integer domain explains the low ambiguity resolution success rate of the MC-LAMBDA method in the first experiment. Data 1 was measured in a particular area with trees, buildings, and other obstructions, which results in signal blockage and multipath errors. Data 2 was collected with good satellite visibility.
Hence, the results of data 2 are remarkably better than those obtained from Data 1. From Table~\ref{tab:suc3}, we see that the proposed method has the best success-rate performance, and the advantage is more significant in challenging environments, which indicates its notable reliability. 

\begin{table}[tbp]
\centering
\caption{Success rate (\%) and RMSE (deg.) based on the experimental data.}
\label{tab:suc3}
\centering
\setlength{\tabcolsep}{1mm}{
\begin{tabular}{|c| c c c| c  c c|}
\hline
   \multirow{2}[2]{*} &\multicolumn{3}{c|}{Success rate}  &\multicolumn{3}{c|}{RMSE} \\
   &MC-LAMBDA &AFM &Proposed  &$\alpha$ &$\beta$ &$\gamma$\\
   \hline
Data 1 &25.85 &76.20 &83.18 &0.13 &0.92 &1.27  \\
\hline
Data 2 &99.28 &97.64 &99.33 &0.03 &0.42  &0.45 \\
\hline
\end{tabular}
}
\end{table}

Fig.~\ref{fig:result1} and Fig.~\ref{fig:result2} shows the distribution of attitude angle estimation error for two experiments. From the plots, it can be readily seen that the proposed method obtains the largest fraction of Euler angle estimates with an error smaller than a given value. This is more visible with data 1. These figures indicate that the proposed approach can provide better attitude angle estimates compared with the AFM-based and the MC-LAMBDA methods.

 								




\section{Conclusion}
\label{sec:conclusion}
A constrained wrapped least-squares (C-WLS) method for attitude determination using multiple GNSS antennas is presented. Unlike existing attitude determination methods, the proposed approach allows the attitude determination problem to be tackled directly--without initially estimating the integer ambiguities--while respecting the antenna array constraints and the integer property of the carrier-phase ambiguities. Given a stated assumption on the double-difference carrier-phase noise, the proposed approach includes an additional constraint on the residual phase to strengthen the optimization model. We solve the C-WLS problem by searching the potential candidates on the unit sphere and obtaining an improved solution through a refinement process. The proposed method targets the estimation of the attitude parameters but inherently returns the integer ambiguities. Simulation and experimental results demonstrate the effectiveness of the proposed approach.

\ifCLASSOPTIONcaptionsoff
  \newpage
\fi

\bibliographystyle{IEEEtran}
\bibliography{IEEEabrv,IEEEfull}

\end{document}